\documentclass[aps,pra,twocolumn,superscriptaddress,10pt,longbibliography,floatfix]{revtex4-2}
\usepackage{babel}
\babelprovide[hyphenrules=english]{en}

\usepackage[utf8]{inputenc}
\usepackage[arrowdel]{physics}
\usepackage{soul}
\usepackage{mathtools}
\usepackage{amsmath}
\usepackage{amssymb}
\usepackage{amsfonts}
\usepackage{bm}
\usepackage{braket}
\usepackage{dsfont}
\usepackage{booktabs}
\usepackage{graphicx}
\graphicspath{{nmi_submission/figures/}}
\usepackage[shortlabels]{enumitem}
\usepackage{float}
\usepackage[unicode=true,bookmarks=false,breaklinks=false,pdfborder={0 0 1},backref=false,colorlinks=false]{hyperref}
\usepackage{xcolor}
\usepackage{xspace}
\usepackage{tikz}

\definecolor{darkblue}{RGB}{0,71,133}      %
\definecolor{darkgreen}{RGB}{0,100,60}     %
\definecolor{royalpurple}{RGB}{88,44,131}    %
\definecolor{teal}{RGB}{0,128,128}         %
\definecolor{plum}{RGB}{142,69,133}        %
\definecolor{rust}{RGB}{183,65,14}         %
\definecolor{slate}{RGB}{47,79,79}         %

\hypersetup{
    colorlinks=true,
    citecolor=black,
    linkcolor=black,
    urlcolor=black
}

\makeatletter
\newcommand{\sidef}[2]{\@namedef{si@ref@#1}{#2}}
\DeclareRobustCommand{\siref}[1]{\@ifundefined{si@ref@#1}{{\ttfamily[?? \detokenize{#1}]}\@latex@warning{Cross-document label `\detokenize{#1}' has no table entry}}{\hyperref[#1]{\@nameuse{si@ref@#1}}}}
\makeatother
\newcommand{\suppref}[1]{\siref{#1}}
\sidef{si:equivariance}{Supplementary Note~1}
\sidef{si:expressivity}{Supplementary Note~2}
\sidef{si:cfi}{Supplementary Note~3}
\sidef{si:qm9}{Supplementary Note~4}
\sidef{si:tsp_benchmark}{Supplementary Note~5}
\sidef{si:trainability}{Supplementary Note~6}
\sidef{si:measurement}{Supplementary Note~7}
\sidef{si:rdm_jacobian}{Supplementary Note~8}
\sidef{si:gate_ablation}{Supplementary Note~9}
\sidef{si:tsp_ablations}{Supplementary Note~10}
\sidef{si:resources}{Supplementary Note~11}
\sidef{lem:neighbourhood}{Supplementary Lemma~1}
\sidef{prop:trainability_si}{Supplementary Proposition~1}
\sidef{prop:rdm_jacobian}{Supplementary Proposition~3}
\sidef{tab:si_equivariance_capped}{Supplementary Table~1}
\sidef{tab:supp_brec_groups}{Supplementary Table~2}
\sidef{tab:qm9_jsweep}{Supplementary Table~3}
\sidef{tab:si_tsp_hparams}{Supplementary Table~4}
\sidef{tab:si_tsp_results}{Supplementary Table~5}
\sidef{tab:si_simcost}{Supplementary Table~6}
\sidef{tab:si_simcost_dk}{Supplementary Table~7}
\sidef{tab:si_trainability}{Supplementary Table~8}
\sidef{tab:supp_cost}{Supplementary Table~9}
\sidef{tab:supp_panel_b}{Supplementary Table~10}
\sidef{tab:rdm_rank}{Supplementary Table~11}
\sidef{tab:gate_ablation}{Supplementary Table~12}
\sidef{tab:tsp5_frozen}{Supplementary Table~13}
\sidef{tab:resource_blocks}{Supplementary Table~14}
\sidef{fig:si_trainability_3panel}{Supplementary Figure~1}
\sidef{fig:si_shots_extrap}{Supplementary Figure~4}
\sidef{fig:gamma_unified}{Supplementary Figure~5}

\usetikzlibrary{arrows.meta, positioning, shapes.geometric, decorations.pathreplacing, backgrounds, fit}

\definecolor{panelA}{RGB}{242,246,240}
\definecolor{panelB}{RGB}{237,244,249}
\definecolor{panelC}{RGB}{249,247,238}
\definecolor{barLow}{RGB}{212,165,160}
\definecolor{barHi}{RGB}{148,188,160}
\definecolor{ceilingClr}{RGB}{180,80,80}
\definecolor{gateLoad}{RGB}{232,170,160}
\definecolor{gateAdj}{RGB}{158,188,214}
\definecolor{gateTrain}{RGB}{162,198,170}
\definecolor{gateMeas}{RGB}{232,200,150}
\definecolor{accentBlue}{RGB}{60,110,160}
\definecolor{accentGreen}{RGB}{75,135,95}
\definecolor{mutedInk}{RGB}{60,60,70}
\definecolor{softGrey}{RGB}{140,140,150}

\usepackage{amsthm}
\usepackage{aliascnt}
\usepackage{thm-restate}

\newtheorem{definition}{Definition}
\newtheorem{proposition}{Proposition}

\providecommand{\Var}{\mathrm{Var}}
\providecommand{\E}{\mathbb{E}}

\newcommand{\RDM}{$1$-RDM\xspace}

\newcommand{\EDIN}{\affiliation{School of Informatics, University of Edinburgh, Edinburgh, UK}}
\newcommand{\FUJITSU}{\affiliation{Fujitsu Research of Europe Ltd., United Kingdom}}
\newcommand{\LIP}{\affiliation{LIP6, CNRS, Sorbonne Universit\'{e}, Paris, France}}
\newcommand{\TERRA}{\affiliation{Terra Quantum GmbH, Germany}}
\newcommand{\TII}{\affiliation{Quantum Research Center, Technology Innovation Institute, Abu Dhabi, UAE}}
\newcommand{\QCWARE}{\affiliation{QC Ware, Palo Alto, USA and Paris, France}}

\makeatletter
\newcommand{\maindef}[2]{\@namedef{main@ref@#1}{#2}}
\newcommand{\main@ref@print}[1]{\@ifundefined{main@ref@#1}{{\ttfamily[?? \detokenize{#1}]}\@latex@warning{Cross-document label `\detokenize{#1}' has no table entry}}{\hyperref[#1]{\@nameuse{main@ref@#1}}}}
\DeclareRobustCommand{\mainref}{\@ifstar{\main@ref@print}{\main@ref@full}}
\newcommand{\main@ref@full}[1]{\main@ref@print{#1} of the main text}
\makeatother
\maindef{thm:equivariance_full}{Theorem~1}
\maindef{thm:jwl_full_informal}{Theorem~2}
\maindef{ssec:results_quantum}{Section~I\,B}
\maindef{ssec:results_expressivity}{Section~I\,D\,1}
\maindef{ssec:results_qm9}{Section~I\,D\,2}
\maindef{ssec:results_tsp}{Section~I\,D\,3}
\maindef{ssec:results_equivariance}{Section~I\,D\,4}
\maindef{ssec:results_trainability}{Section~I\,D\,5}
\maindef{fig:cfi_discrimination}{Figure~2}
\maindef{fig:tsp_demonstration}{Figure~2}
\maindef{fig:trainability}{Figure~4}

\newtheorem{silemma}{Supplementary Lemma}
\newtheorem{siproposition}{Supplementary Proposition}
\newtheorem*{maintheorem}{\mainref{thm:equivariance_full}}
\newtheorem*{mainjwltheorem}{\mainref{thm:jwl_full_informal}}

\begin{document}

\title{Scalable Message-Passing Quantum Graph Neural Networks \\ in the Weisfeiler–Leman Hierarchy}

\author{Snehal Raj}
\email[Corresponding author: ]{snehal.raj@lip6.fr}
\LIP\QCWARE
\author{Brian Coyle}\EDIN\FUJITSU
\author{L\'eo Monbroussou}\EDIN\TERRA
\author{Andr\'e\ J.\ Ferreira-Martins}\TII\LIP
\author{Renato\ M.\ S.\ Farias}\TII
\author{Elham\ Kashefi}\LIP\EDIN

\begin{abstract}
Graphs provide a natural language for relational data in chemistry, biology and optimisation. Graph neural networks (GNNs) have driven much of the recent progress in learning from such data through message passing, a single primitive that generalises convolution and attention. Quantum counterparts have been proposed, but with limited connection to message passing and few guarantees on performance or scalability. More broadly, the trainability of variational quantum circuits is a recognised bottleneck for their wide applicability, and pre-training has emerged as one way to address it. Yet for a quantum model to be useful, it must offer expressivity guarantees along with demonstrable scalability. Here we show how a quantum graph neural network can be built to perform message passing, to be permutation equivariant, and to sit at a chosen level of the Weisfeiler--Leman hierarchy, the standard measure of how finely a model can tell graphs apart. We show that, as for classical GNNs, the training can be done first on small graph instances, allowing for a pre-training that can mitigate usual training issues, and its output can be read out at a cost that stays low as the graph grows. We validate the framework in large-scale simulations of up to 56 qubits across three datasets, on synthetic graphs that ordinary message passing cannot separate, on molecular property prediction, and on the travelling salesperson problem. Our framework opens a path for near-term quantum algorithms with theoretical guarantees and practical scalability, bringing the principles of graph learning into quantum circuit design.
\end{abstract}

\maketitle

Graph neural networks (GNNs) have become a central tool in modern machine learning, driving progress in molecular property prediction and drug discovery~\cite{gilmer2017neural}, protein structure modelling as realised in AlphaFold~\cite{AlphaFold}, and combinatorial optimisation on graphs~\cite{joshi2019efficient, Kool2019}. Their success rests on \textit{message passing}, the rule by which each node updates from its neighbours~\cite{gilmer2017neural, Xu2019}. Quantum versions have been proposed in several forms~\cite{Ceschini2024, verdon2019quantum, Ai2024, Ryu2023, hu2022design, zheng2021quantum, mernyei2022equivariant}, and like their classical counterparts they can be made permutation equivariant~\cite{Skolik2023, schatzki2024equivariant, nguyen2024theory}. Most of them, however, do not actually carry out message passing: they derive the circuit topology from the graph, placing gates along the graph's edges so that the circuit layout mirrors the graph structure~\cite{verdon2019quantum, zheng2021quantum, hu2022design}, while the step that defines a graph neural network, aggregating over each node's neighbours and updating its state, is left to a classical readout rather than run inside the circuit, because performing it coherently is expensive. These models also face the trainability difficulty common to variational circuits, whose gradients can vanish as the circuit grows~\cite{larocca_barren_2025, mcclean2018barren, ragone2024lie}. Therefore, it is important to show that a quantum graph model can be trained and scaled to large graphs.

Message passing is what makes a model a graph neural network, and a unifying principle of modern deep learning: on regular domains, this single primitive subsumes convolutional, recurrent, and attention-based architectures as special cases~\cite{gilmer2017neural, bronstein2021geometric}. Because a graph carries no intrinsic ordering of its nodes, any graph operation must be permutation equivariant by construction, so relabelling the nodes relabels the output the same way~\cite{maron2019invariant,bronstein2021geometric}. What limits such a model is how finely it can separate structurally distinct graphs, and the standard yardstick for this is the Weisfeiler--Leman (WL) hierarchy, a sequence of colour-refinement tests of increasing power for distinguishing non-isomorphic graphs~\cite{morris2019weisfeiler}. Ordinary message passing reaches only the first level of this hierarchy, the $1$-WL test, so increasing depth or width never separates a pair of graphs that $1$-WL cannot tell apart~\cite{Xu2019,morris2019weisfeiler}. This ceiling has concrete consequences, since graphs indistinguishable to $1$-WL can correspond to chemically distinct molecules and to substructures the network is therefore unable to count~\cite{chen2020can}.

In this work, we introduce a quantum graph neural network that faithfully carries out message passing, is permutation equivariant, and sits in the Weisfeiler--Leman hierarchy at a level set by the model, rising above the $1$-WL ceiling (\autoref{fig:framework}). We prove these properties for an instantiation built from subspace-preserving circuits. Confining the dynamics to a fixed subspace keeps the model trainable, at the cost of a classical simulation with polynomial overhead set by the size of the subspace~\cite{ragone2024lie, Fontana2024, cerezo2025does, monbroussou2025trainability, recioarmengol2025train, bako2025fermionic, rudolph2023synergistic}. As is done with classical GNNs, we pre-train on small graphs and apply the model to larger instances. This avoids cost concentration, since the trainable dynamics of the model do not scale directly with the size of the graph. We demonstrate trainability and scaling directly, in numerical simulations of up to $56$ qubits on three datasets (\autoref{fig:results}): Cai--F\"urer--Immerman graphs, a synthetic benchmark of pairs that ordinary message passing cannot separate, used to validate the expressivity claim; QM9 molecular property prediction; and the Euclidean travelling salesperson problem.

We present our framework and its guarantees in Sec.~\ref{sec:results}, first by introducing the general framework to define GNNs along with the theoretical properties that make them useful. We then show how our framework, when instantiated with subspace-preserving circuits, respects these properties and how it stays scalable through the pre-training strategy and its readout. Sec.~\ref{sec:numerics} validates the model numerically on three datasets. Cai--F\"urer--Immerman graphs~\cite{cai1992optimal}, the standard benchmark for the Weisfeiler--Leman test, confirm the expressivity result on graph pairs that no message-passing network can tell apart. QM9 regression~\cite{ramakrishnan2014quantum} carries the same Weisfeiler--Leman gain over to real molecular property prediction. Euclidean travelling-salesperson instances up to TSP-50 ($56$ qubits) exercise equivariance, trainability, and the train-small, deploy-large pipeline.

\section{Results}\label{sec:results}

We begin by recalling what defines a graph neural network and the theoretical properties that make such a model useful. We then present a general framework for building graph neural networks with these properties and instantiate it with subspace-preserving quantum circuits~\cite{kerenidis2021quantum, Landman2022, cherrat2023quantum, monbroussou2025trainability, monbroussou_subspace_2025, Raj2025}. We prove that this quantum instantiation realises these properties, and show how it stays scalable and trainable as the graph grows.

\subsection{Graph Neural Network framework}
\label{ssec:results_framework}

\begin{figure*}[!t]
\centering
\includegraphics[width=\linewidth]{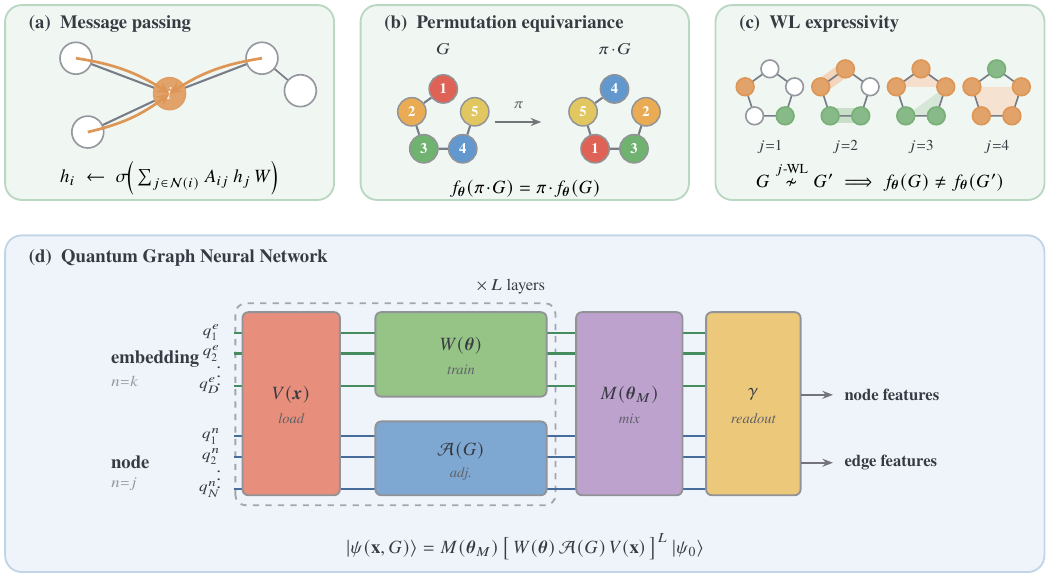}
\caption{\textbf{Graph-learning properties and the quantum graph neural network.} (a)~\textbf{Message passing.} Each node updates from its neighbours, $h_i\leftarrow\sigma\!\big(\sum_{j\in\mathcal{N}(i)}A_{ij}\,h_j\,W\big)$, so after $L$ rounds a node aggregates its $L$-hop neighbourhood. (b)~\textbf{Permutation equivariance.} Relabelling the input nodes by $\pi$ relabels the output by the same $\pi$, $f_{\boldsymbol{\theta}}(\pi\!\cdot\!G)=\pi\!\cdot\!f_{\boldsymbol{\theta}}(G)$. (c)~\textbf{Weisfeiler--Leman expressivity.} The node-register particle number $j$ sets the level of the set-based $j$-Weisfeiler--Leman hierarchy that the model reaches: at generic parameters{\color{black}, with the encoder, pooling and depth of \autoref{thm:jwl_full_informal},} it gives two graphs separated by $j$-WL distinct outputs, $G\not\sim_{j\text{-WL}}G'\Rightarrow f_{\boldsymbol{\theta}}(G)\neq f_{\boldsymbol{\theta}}(G')$, exceeding the $1$-WL ceiling for $j\geq 3$. (d)~A \emph{node register} ($N$ qubits at particle number $j$) and a \emph{feature-embedding register} ($D$ qubits at particle number $k$). A loader $V(\mathbf{x})$ encodes the graph, an equivariant adjacency $\mathcal{A}(G)$ acts on the node register and a trainable evolution $W(\boldsymbol{\theta})$ on the embedding register, alternating over $L$ layers; a joint-register mixing layer $M(\boldsymbol{\theta}_M)$ couples the registers, and a readout $\gamma$ returns the node and edge features.}
\label{fig:framework}
\end{figure*}

A graph neural network processes a graph $G$ on $N$ nodes. An encoder $V(G)$ collects the $D_F$-dimensional node features into a matrix $X \in \mathbb{R}^{N \times D_F}$. The connectivity is carried by an adjacency operation $\mathcal{A}(G) \in \mathbb{R}^{N \times N}$. Message passing then refines these features over $L$ rounds: with $H^{(0)} = X$, each round transforms $H^{(\ell)}$ into $H^{(\ell+1)}$ by alternating $\mathcal{A}(G)$, which aggregates each node's neighbours along the edges, with a trainable evolution $W(\boldsymbol{\theta}) \in \mathbb{R}^{D_F \times D_F}$ acting on the feature space, followed by a nonlinearity $\sigma$.

In order to produce useful GNNs, it is important to respect several theoretical properties that the literature has identified~\cite{gilmer2017neural, bronstein2021geometric, maron2019invariant} (\autoref{fig:framework}(a)--(c)), and we explain how they motivate a proper processing of the input graph data.

\begin{definition}[Message passing]\label{def:message_passing}
    One layer maps $H^{(\ell)} \mapsto H^{(\ell+1)} = \sigma\!\big(\mathcal{A}(G)\,H^{(\ell)}\,W({\boldsymbol{\theta}})\big)$, with $\sigma$ a suitable nonlinearity; after $L$ layers, the embedding of node $i$ encodes aggregated features from its $L$-hop neighbourhood in $G$.
\end{definition}

Message passing is the operational core of nearly every graph neural network, building each node's representation by repeatedly aggregating its neighbours and updating itself, so the model exploits the graph's connectivity rather than the node features alone~\cite{gilmer2017neural, Xu2019}. It is also the unifying principle of modern deep learning, with convolution, recurrence, and attention all recovered as message passing on grids, sequences, and fully connected graphs~\cite{bronstein2021geometric, Velickovic2018}.

\begin{definition}[Exact permutation equivariance]\label{def:perm_equiv}
    Let $f$ be the function computed by a GNN. It is permutation equivariant if $f(\pi\!\cdot\!G,\, \pi\!\cdot\!X) = \pi\!\cdot\! f(G,X)$ for every relabelling $\pi \in S_N$ of the nodes.
\end{definition}

Because the nodes of a graph carry no intrinsic order, permutation equivariance is essential for the model to return the same predictions for the same graph under different labelling~\cite{maron2019invariant, bronstein2021geometric}. Encoding the symmetry into the architecture rather than learning it from data restricts the model to label-consistent functions, which lowers sample complexity and improves generalisation~\cite{nguyen2024theory, schatzki2024equivariant}.

\begin{definition}[Weisfeiler--Leman test]
    The $j$-Weisfeiler--Leman ($j$-WL) test assigns each $j$-subset $S$ of the nodes a colour $c^{t}(S)$, refined over steps $t = 0, 1, 2, \dots$ from the initial colour $c^{0}(S)$ given by the isomorphism type of the subgraph induced by $S$. Each step updates
    \begin{equation}
        c^{t+1}(S) = \mathrm{HASH}\!\Big( c^{t}(S),\ \{\!\!\{\, c^{t}(S') : |S \triangle S'| = 2 \,\}\!\!\} \Big),
    \end{equation}
    where $\mathrm{HASH}$ is an injective relabelling, $\{\!\!\{\cdot\}\!\!\}$ denotes a multiset, and $S \triangle S'$ is the symmetric difference, so the multiset ranges over the $j$-subsets $S'$ obtained from $S$ by replacing a single node. Once the colours stabilise, two graphs are distinguished when their multisets of subset colours differ~\cite{morris2019weisfeiler}.
\end{definition}

\begin{definition}[Weisfeiler--Leman expressivity]\label{def:jwl_expressivity}
    An architecture has $j$-WL expressivity when it matches the Weisfeiler--Leman test, returning identical outputs for any two graphs that $j$-WL leaves indistinguishable and different outputs whenever $j$-WL separates them.
\end{definition}

The Weisfeiler--Leman hierarchy is the standard measure of a model's distinguishing power, and ordinary message passing reaches only its first level, $1$-WL~\cite{Xu2019, morris2019weisfeiler}. This ceiling has concrete consequences, since $1$-WL cannot separate graphs that differ in higher-order structure, including molecules with distinct properties and substructures such as cycles that the network cannot count~\cite{chen2020can}; lifting it is a central aim of higher-order graph learning~\cite{maron2019invariant, morris2020weisfeiler}.

\subsection{Quantum Graph Neural Network}
\label{ssec:results_quantum}

Here we give a general framework for constructing QGNNs from four building blocks: a data loader $V$, an equivariant adjacency layer $\mathcal{A}(G)$, a trainable evolution $W(\boldsymbol{\theta})$, and a joint-register mixing layer $M(\boldsymbol{\theta}_M)$ (\autoref{fig:framework}(d)). We also instantiate this framework with subspace-preserving quantum circuits~\cite{kerenidis2021quantum, Landman2022, cherrat2023quantum, monbroussou2025trainability, monbroussou_subspace_2025, Raj2025, mathur2025bayesian, coyle2025training, mathur_scalable_2026, jain_quantum_2024}, and show how the resulting architecture respects the properties of \autoref{ssec:results_framework}. We detail each block below:

\paragraph{Hierarchical data loader.} To reproduce the framework's structure, we split the qubits into two registers. The $N$-qubit \emph{node register} indexes the graph's nodes at Hamming weight $j$, spanning the subspace $\mathcal{H}_N^j$ of dimension $\binom{N}{j}$: its basis states are individual nodes at $j{=}1$ and $j$-subsets of nodes (edges, triangles, and higher-order groups) at larger $j$. The $D$-qubit \emph{embedding register} holds the corresponding features at Hamming weight $k$, in $\mathcal{H}_D^k$ with $D_F = \binom{D}{k}$. Indexing the node basis by $j$-subsets $T$, the loader $V$ prepares
\begin{equation}\label{eq:Graph_Encoding}
    \ket{\mathbf{x}} = \sum_{T} \sum_{f=1}^{D_F} x^{(T)}_f \,\ket{e^{\text{node}}_T} \otimes \ket{e^{\text{embedding}}_f},
\end{equation}
with $e^{\text{node}}_T \in \{0,1\}^N$ of Hamming weight $j$, $e^{\text{embedding}}_f \in \{0,1\}^D$ of Hamming weight $k$, and $x^{(T)}$ the features of subset $T$ (the node features at $j{=}1$). It is realised by the controlled-Givens particle-number encoder of~\cite{Farias2025}, and the reuploading layers below apply it to an arbitrary state in $\mathcal{H}_N^j \otimes \mathcal{H}_D^k$.

\paragraph{Equivariant adjacency $\mathcal{A}(G)$.} The adjacency layer applies Givens rotations between the node-qubit pairs that form an edge in $G$, encoding the graph connectivity. {\color{black}We show in \autoref{thm:equivariance_full} that the layer is equivariant.}

\paragraph{Trainable evolution $W(\boldsymbol{\theta})$.} Acts on the embedding register only and is shared across the node-register components. Many gate sets are admissible here, with different trainability and expressivity according to their dynamical Lie algebra. Our instantiation uses subspace-preserving gates: a Cayley parametrisation of $\mathrm{SO}(D)$ lifted to $\mathrm{SO}(D_F)$ by the $k$-th compound representation $C^{(k)}(\cdot)$~\cite{cherrat2023quantum, Kerenidis2022}, with non-compound alternatives also available~\cite{monbroussou2025trainability, Farias2025}. We choose this set for its favourable trainability dynamics.

\paragraph{Reuploading and nonlinearity.} The model is built from $L$ layers, and the input features are re-uploaded through the amplitude encoder $V$ at each layer. This transformation creates non-linear features in the amplitudes of the quantum state and makes the overall map from input features to output nonlinear~\cite{perez-salinas2020data}. Fermionic ladder operators on each register follow from the Jordan--Wigner mapping under the same lexicographic qubit ordering; the embedding-register $1$-RDM used as the readout in \autoref{ssec:results_trainability} is $\gamma_{pq} = \bra{\psi}\,a_p^\dagger\,a_q\,\ket{\psi}$ for $p,q\in[D]$.

\paragraph{Joint-register mixing $M(\boldsymbol{\theta}_M)$.} A final layer mixes the two subspaces. It applies Givens rotations between qubit pairs within and across the registers, taking the state into the subspace of $N+D$ qubits at Hamming weight $j+k$. The layer has three kinds of gates: between node qubits, between embedding qubits, and across the registers.

\subsection{Theoretical guarantees}

In this section, we present the theoretical guarantees of the QGNN framework and its subspace-preserving instantiation. We first show that the instantiation realises, precisely, each of the three properties defined in \autoref{ssec:results_framework}: message passing, exact permutation equivariance, and Weisfeiler--Leman distinguishing power. We then show that the same instantiation is scalable: it can be pre-trained on small graphs and deployed on larger ones; similar to the usual training pipeline for classical GNNs. Our experiments show that these models remain trainable on large instances and that their readout cost remains low.

\subsubsection{Expressivity guarantees}

\paragraph{Message passing.} The model performs message passing inside the circuit, realising the layer map of \autoref{def:message_passing}. The adjacency $\mathcal{A}(G)$ acts on the node register and routes each node's amplitude along the edges; the trainable evolution $W(\boldsymbol{\theta})$ acts on the embedding register and updates the per-node features; and re-uploading the loader at each layer supplies the nonlinearity. {\color{black}The circuit blocks use quantum operations; the expressivity construction below also uses a shared classical encoder.} The construction differs from a classical $1$-GNN in two ways. First, the aggregation: a $1$-GNN sums over each node's neighbourhood multiset~\cite{Xu2019, morris2019weisfeiler}, while the QGNN uses the deterministic real-orthogonal map $\mathcal{A}(G)$ on a fixed subspace of dimension $\binom{N}{j}$. Second, the connectivity: a classical network uses a hard adjacency mask, while here every pair of node qubits is coupled, and the graph enters through the rotation angles. {\color{black}For the expressivity result, Johnson mixing $\exp(i\alpha A_J)$ uses the generator $(A_J)_{S,S'}=1$ when $S$ and $S'$ differ by a single swap, the one-swap neighbourhood that defines set-based $j$-WL~\cite{morris2019weisfeiler} (\siref{lem:neighbourhood}).}

\paragraph{Exact permutation equivariance.} {\color{black}Relabelling the input nodes permutes the output node and edge features by the same permutation at every parameter value.} This is a consequence of the framework's two-register structure. {\color{black}The trainable evolution acts only on the embedding register and is shared across the node-register components.} Together, these give exact $S_N$-equivariance ($S_N$ the symmetric group of all $N!$ possible relabelling of the nodes) without tying any parameters together. Earlier equivariant graph models, classical and quantum alike, instead impose the symmetry by sharing or tying parameters across the symmetric group~\cite{maron2019invariant, bronstein2021geometric, Skolik2023}, which constrains the function class.

\begin{restatable}[Exact $S_N$-equivariance]{theorem}{ExactSNequivariance}
\label{thm:equivariance_full}
{\color{black}
Let $\mathbf{x}=(\mathbf{X},A)$ contain the node features and adjacency matrix.
Under the assumptions and notation in \suppref{si:equivariance},
for every parameter assignment $\boldsymbol{\theta}$ and every $\pi\in S_N$,
\begin{equation}
f_{\boldsymbol{\theta}}
  (\pi\cdot\mathbf{x};\pi\cdot\mathcal{K})
  =\pi\cdot f_{\boldsymbol{\theta}}(\mathbf{x};\mathcal{K}).
\end{equation}
Here $\pi$ permutes the node or subset outputs, and acts trivially on a
symmetrically pooled graph output. The trainable angles need no tying
within or across the three joint-mixing blocks.
}
\end{restatable}

We prove \autoref{thm:equivariance_full} in \suppref{si:equivariance}, and \autoref{ssec:results_equivariance} measures its effect on data efficiency (\autoref{fig:equivariance_ablation}).

\paragraph{Expressivity beyond the $1$-WL ceiling.} A message-passing network distinguishes two graphs only when the $1$-Weisfeiler--Leman ($1$-WL) test does, a ceiling that bounds every standard architecture. The node-register particle number $j$ sets how much graph structure the model resolves: at generic parameters, for $2 \leq j \leq 4$, it matches the set-based $j$-WL test exactly (\autoref{def:jwl_expressivity}), {\color{black}with the encoder and pooling below at sufficient depth,} and from $j{=}3$ it separates graphs that $1$-WL, and hence every message-passing network, cannot. {\color{black}Constructive examples also establish separation beyond $1$-WL with the experimental encoder (\suppref{si:expressivity}).}

\begin{samepage}
\begin{restatable}[$j$-WL expressivity, informal]{theorem}{jWLExpressivityInformal}\label{thm:jwl_full_informal}
{\color{black}
For $2\leq j\leq4$, under the assumptions in \suppref{si:expressivity},
the graph output at generic parameters and sufficient depth satisfies
\begin{equation*}
f_{\boldsymbol\theta}(G)=f_{\boldsymbol\theta}(G')
\quad\Longleftrightarrow\quad G\sim_{j\text{-WL}}G',
\end{equation*}
where $j$-WL denotes stable set-based refinement.
Subsets with the same stable colour retain equal features at every iteration.
For $j\geq3$, this construction can separate graphs that $1$-WL leaves
indistinguishable.
}
\end{restatable}
\end{samepage}

The initialisation {\color{black}carries} the graph's structure into the model: two copies of the model with the same {\color{black}generic} parameters {\color{black}and sufficient depth}, run on two graphs that $j$-WL tells apart, still produce different outputs, which is why the choice of initialisation is the essential ingredient. The construction works for groups of up to four nodes; \suppref{si:expressivity} gives a concrete choice, proves the result, and explains why larger groups need more information. \autoref{ssec:results_expressivity} validates the ascent on Cai--F\"urer--Immerman families (\autoref{fig:cfi_discrimination}).

\subsubsection{Scalability}

A central obstacle to scaling variational quantum algorithms is the barren plateau phenomenon, where gradients vanish as the circuit grows and prevent training~\cite{mcclean2018barren, larocca_barren_2025}. The phenomenon is tied to the size of the space the circuit's dynamics explore: it arises when those dynamics reach an exponentially large space, and is mitigated when they stay within a smaller, structured subspace~\cite{ragone2024lie, Fontana2024, monbroussou2025trainability, cerezo2025does}. The two-register QGNN works inside such a subspace; we recall the relevant trainability results in \suppref{si:trainability}.

\begin{proposition}[GNN pre-training]\label{prop:GNN_pre_training}
A GNN can be pre-trained on small graph instances: the trainable evolution $W(\boldsymbol{\theta})$ acts only on the embedding register, whose size does not typically grow with the graph, while the adjacency layer requires no training.
\end{proposition}

This mirrors standard practice for classical graph networks, whose trainable weights act on the feature channels at a width chosen independently of the number of nodes, so a model trained on small graphs is routinely run on much larger ones, as in graph-based combinatorial optimisation~\cite{joshi2019efficient, Kool2019}. Our architecture inherits the property because the trainable evolution lives on the embedding register, kept separate from the growing node register. We use it through a \textbf{transferable initialisation}, in which parameters fitted at one graph size initialise the next; a similar train-small, deploy-large idea has recently been proposed for variational quantum models to ease their training~\cite{recioarmengol2025train, rudolph2023synergistic, bako2025fermionic}.

\paragraph{Trainability.} Our technique is heuristic, but the numerical evidence is favourable. Although the state after the mixing layers lies in the joint subspace of $N+D$ qubits at Hamming weight $j+k$, the gradients we observe do not follow that subspace dimension, and the transferable initialisation improves the trend further (\autoref{ssec:results_trainability}). A supporting trainability argument, under stated assumptions, is given in \suppref{si:trainability} (\siref{prop:trainability_si}).

\paragraph{Readout at polynomial cost.} To ensure scalability, the trained model must also be easy to read out for deployment.
\begin{proposition}[Polynomial-cost readout]\label{prop:rdm_readout}
The node and edge features are the entries of the embedding-register one-particle density matrix $\gamma_{pq}=\bra{\psi}a_p^\dagger a_q\ket{\psi}$ ($p,q\in[D]$), a $D\times D$ matrix over the $D$ embedding modes with trace $k$. It is read in two steps: the node register is projected onto a single node, and the embedding register of the conditional state is measured. This matrix can be estimated to accuracy $\varepsilon$ in $O(D^3/\varepsilon^2)$ measurement shots, either by a deterministic Hartree--Fock readout~\cite{arute2020hartree} or by a matchgate shadow~\cite{wan2023matchgate} on the embedding register. The cost is set by the $D$-qubit feature register, not the number of nodes; the per-node projection multiplies it by $1/P_{\rm proj}$, the inverse success probability of the node-register projection, which is polynomial in $N$.
\end{proposition}
We prove this in \suppref{si:measurement}. Computational-basis measurements alone recover only the diagonal of this matrix (\autoref{ssec:results_trainability}).

Together, these let the model be trained on small graphs and run on larger ones.

\subsection{Numerical validations}\label{sec:numerics}

We test the framework's guarantees with large-scale numerical simulations. The expressivity, equivariance, trainability, and readout properties established as theorems and propositions in the previous subsections are validated here at scale: the runs use multiple random seeds, reach up to $56$ qubits, and cover both synthetic graphs and real-world data, showing that these properties hold and are useful in practice. We work with three datasets. The first, Cai--F\"urer--Immerman (CFI) graphs~\cite{cai1992optimal}, is a synthetic test of the $j$-WL ascent of \autoref{thm:jwl_full_informal}. The second, QM9~\cite{ramakrishnan2014quantum}, carries the same ascent over to a real molecular-property task. The third, the Euclidean travelling-salesperson problem~\cite{joshi2019efficient, Kool2019}, runs the full pipeline end-to-end at $j{=}1$, and also carries our ablations of permutation equivariance and trainability.

\subsubsection{Distinguishing graphs beyond $1$-WL: CFI families}
\label{ssec:results_expressivity}

On Cai--F\"urer--Immerman (CFI) graphs, the QGNN's distinguishing power rises with the node-register particle number $j$ exactly at the level {\color{black}\autoref{thm:jwl_full_informal}} predicts. CFI graphs are pairs of non-isomorphic graphs built so that the $1$-Weisfeiler--Leman ($1$-WL) test, and every standard message-passing network with it, cannot tell them apart~\cite{cai1992optimal}; separating a pair requires going beyond $1$-WL. Because the node register at particle number $j$ encodes $j$-subset interactions, the model's distinguishing power is {\color{black}bounded above by} set-based $j$-WL~\cite{morris2019weisfeiler}, and crossing the $1$-WL ceiling needs $j{\geq}3$. We use the CFI($K_3$) pair on $N{=}18$ vertices, separable at $j{=}3$ and not below.

The test accuracy jumps from chance to $100\%$ in all five seeds exactly at the predicted particle number $j{=}3$, and stays at chance below it (\autoref{fig:cfi_discrimination}(a)); the protocol is in \suppref{si:cfi}.

A $1$-WL graph isomorphism network at matched parameter budgets stays at chance on the pair, as predicted. Exact set-$j$-WL on the BREC benchmark~\cite{wang2024brec} gives the ceiling against which a model at level $j$ is measured (\suppref{si:cfi}, \siref{tab:supp_brec_groups}).

\begin{figure*}[t]
\centering
\includegraphics[width=\linewidth]{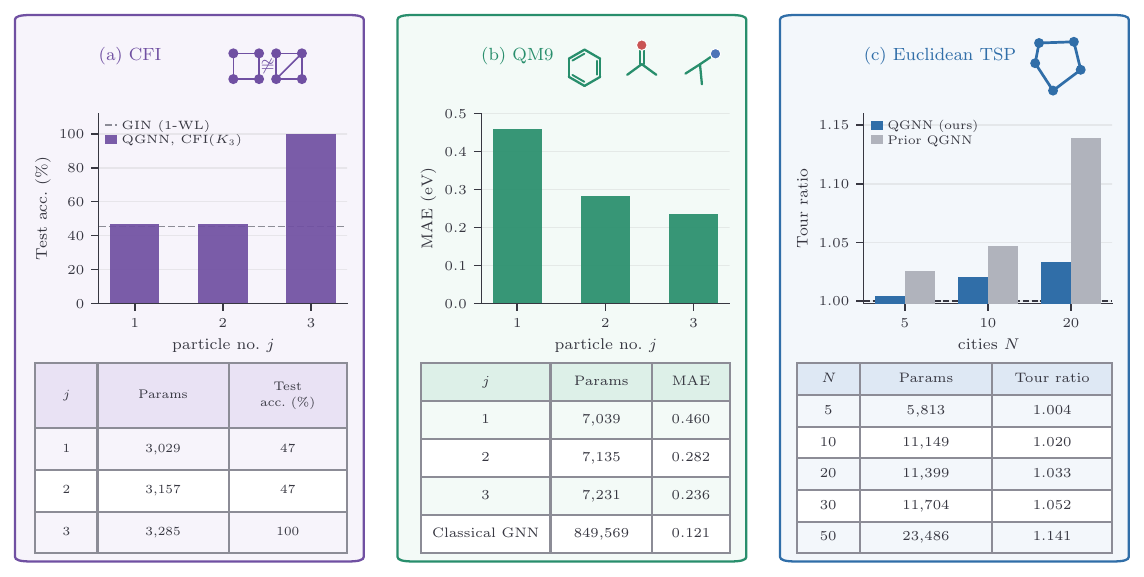}
\caption{\textbf{Numerical validation across three graph-learning tasks.} \textbf{(a)~CFI graph discrimination.} Test accuracy on the 18-vertex Cai--F\"urer--Immerman pair CFI($K_3$), which $1$-WL message passing cannot separate, as the node-register particle number $j$ increases; accuracy is at chance for $j{\le}2$ and reaches $100\%$ in all five seeds exactly at the predicted level $j{=}3$, while a $1$-WL GIN at matched parameter budgets stays at chance (dashed line; \suppref{si:cfi}). \textbf{(b)~QM9 molecular property prediction.} Mean absolute error on the HOMO--LUMO gap over the full QM9 set ($130{,}831$ molecules) falls monotonically with $j$ at near-constant parameter count; a classical GNN, trained in its own pipeline, reaches a lower error at roughly $120\times$ the parameter count (\suppref{si:qm9}). \textbf{(c)~Euclidean travelling salesperson} at $j{=}1$: tour ratio ($1.0$ is optimal) versus the number of cities $N$, against a prior QGNN, the equivariant quantum circuit of Skolik et al.~\cite{Skolik2023}, with the table extending to $N{=}30, 50$ (published quantum-baseline comparisons in \suppref{si:tsp_benchmark}).}
\label{fig:results}\label{fig:cfi_discrimination}\label{fig:qm9_jsweep}\label{fig:tsp_demonstration}
\end{figure*}

\subsubsection{Molecular property prediction (QM9)}
\label{ssec:results_qm9}

QM9 is a standard benchmark in quantum chemistry: $130{,}831$ small organic molecules, each with up to nine heavy atoms (carbon, nitrogen, oxygen, fluorine) and a set of properties computed by density-functional theory~\cite{ramakrishnan2014quantum}. We predict the HOMO--LUMO gap, the energy difference between a molecule's highest occupied and lowest unoccupied orbitals, a quantity that governs much of its chemical reactivity and is a common target for learned property prediction. We ask whether raising the node-register particle number $j$, {\color{black}which raises the encoded subset order}, lowers the error on this real task.

The only quantity changed between runs is $j$: the embedding, the circuit depth, the head, and the optimiser are all held fixed (Methods~M1).

The error falls monotonically as $j$ rises, from $0.460$~eV at $j{=}1$ to $0.282$~eV at $j{=}2$ and $0.236$~eV at $j{=}3$ (\autoref{fig:qm9_jsweep}(b), \siref{tab:qm9_jsweep}); a five-seed run on a smaller subset gives the same ordering (\suppref{si:qm9}). The parameter count is almost unchanged across the sweep, so the gain comes from the added node-register resolution rather than a larger model. A classical message-passing network~\cite{morris2019weisfeiler}, trained in its own pipeline, reaches a lower error, but with about $120\times$ more parameters.

\subsubsection{Euclidean travelling salesperson at \texorpdfstring{$j{=}1$}{j=1}}
\label{ssec:results_tsp}

The travelling-salesperson problem entails finding the shortest closed tour through $N$ cities placed in the plane, and it is NP-hard. Learning to solve it is an active area for graph neural networks, with both classical~\cite{joshi2019efficient, Kool2019} and quantum~\cite{Skolik2023} architectures developed for the task. We follow the Euclidean TSP benchmark setting of Skolik et al.~\cite{Skolik2023}.

Each instance is solved end-to-end at $j{=}1$ (Methods~M1), from $N{=}5$ to $50$ cities, the largest a $56$-qubit simulation (\siref{tab:si_simcost}).

Tour quality is the ratio of the produced tour length to the optimal, where $1.0$ is optimal (\autoref{fig:tsp_demonstration}(c), \siref{tab:si_tsp_results}; protocol in \suppref{si:tsp_benchmark}). The mean tour ratios are $1.004$, $1.020$, $1.033$, $1.052$ and $1.141$ for $5$, $10$, $20$, $30$ and $50$ cities, respectively. These ratios are lower than the published quantum-baseline results at the corresponding sizes~\cite{Skolik2023,Teo2025} (\suppref{si:tsp_benchmark}). At matched parameter counts, a classical graph convolutional network reaches a lower binary cross-entropy. Two ablations are in \suppref{si:gate_ablation} and \suppref{si:tsp_ablations}.

\subsubsection{Permutation equivariance}
\label{ssec:results_equivariance}

\begin{figure}[h!]
\centering
\includegraphics[width=\linewidth]{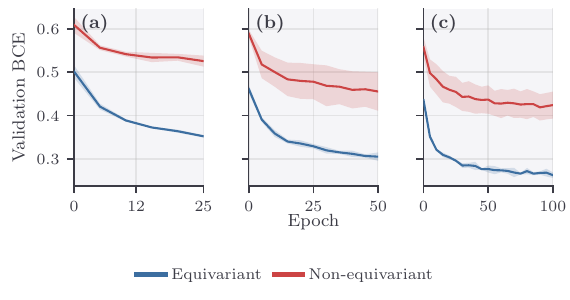}
\caption{\textbf{Equivariance helps the model learn from less data.} Training curves on a $5$-city travelling-salesperson edge-prediction task at three fixed training-set sizes ((a)~$n{=}500$, (b)~$n{=}1000$, (c)~$n{=}2000$). The $y$-axis is the validation prediction loss (lower is better). The blue curve is the equivariant model ($5{,}773$ parameters); the red curve is the same architecture with the canonical edge ordering of $\mathcal{A}$ replaced by a fixed ordering, which breaks the symmetry ($5{,}813$ parameters). Both arms train for the same number of epochs at every training-set size, with no early stopping. Shaded bands are $\pm$ one standard deviation across three random seeds. Equivariance gives a lower loss at every training-set size and at every epoch within each panel; the values at the capped epochs are in \suppref{si:equivariance}.}
\label{fig:equivariance_ablation}
\end{figure}

Built-in permutation symmetry makes the model more data-efficient: it removes the need to relearn the same rule under every node relabelling. The standard route to $S_N$-equivariance in QGNNs ties parameters across the symmetric group~\cite{Skolik2023,verdon2019quantum,liao2024graph}, which costs expressivity inside the equivariant class. {\color{black}Our two-register construction is exactly equivariant without tying any parameters (\autoref{thm:equivariance_full}). The trainable evolution $W(\boldsymbol{\theta})$ acts identically on each node's embedding register.}

On TSP-5 edge prediction the equivariant model reaches a lower validation binary cross-entropy than a variant with a fixed edge ordering, matched to within $40$ parameters, at every training-set size: by $32\%$, $33\%$ and $38\%$ at $n{=}500$, $1000$ and $2000$ (\autoref{fig:equivariance_ablation}, \siref{tab:si_equivariance_capped}). The gap is consistent with the hypothesis-class reduction expected from $S_N$-equivariance~\cite{nguyen2024theory,schatzki2024equivariant}.

\begin{figure*}[t!]
\centering
\begin{minipage}[c]{0.49\textwidth}
\centering
\includegraphics[width=\linewidth]{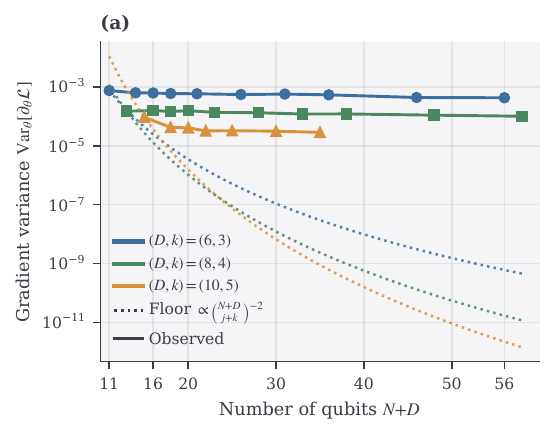}\\
{\footnotesize (a) Per-parameter gradient variance vs.\ qubit count}
\end{minipage}\hfill
\begin{minipage}[c]{0.49\textwidth}
\centering
\includegraphics[width=\linewidth]{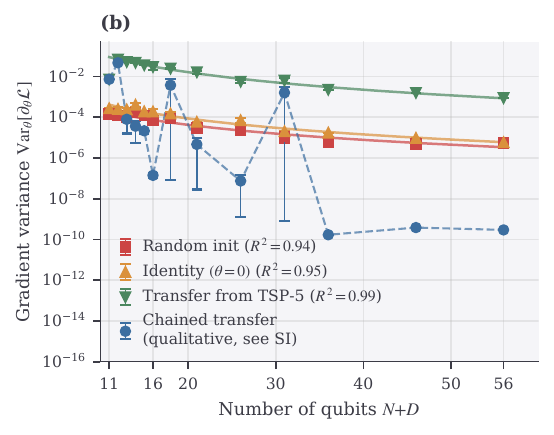}\\
{\footnotesize (b) Transferable initialisation across graph sizes}
\end{minipage}
\caption{\textbf{Polynomial gradient variance and transferable initialisation.} (a)~Per-parameter gradient variance of our QGNN at three half-filling operating points $(D,k)\in\{(6,3),(8,4),(10,5)\}$, $j{=}1$, against the qubit count $N{+}D$ up to $58$ qubits ($27$ cells; mean over $200$ random initialisations per point). Dotted lines: a $\binom{N{+}D}{j{+}k}^{-2}$ reference, anchored and colour-matched to each curve, which every measured curve stays well above. (b)~Gradient signal at the start of training under four initialisation regimes: random parameters (squares), identity ($\boldsymbol{\theta}{=}0$, triangles), single-step transfer from a TSP-5-trained QGNN (down triangles), and chained transfer from $N{-}1$ to $N$ (circles, joined by a dashed line without a fit; see \suppref{si:trainability}). For the three fitted regimes a log--log regression confirms \emph{polynomial} (power-law) decay rather than the exponential decay of a barren plateau ($R^2=0.94,\,0.95,\,0.99$ for the random, identity, and single-step transfer fits). Sweep ranges and the gate-group and shot-complexity breakdowns are in \suppref{si:trainability}.}
\label{fig:trainability}
\end{figure*}

\subsubsection{Trainability and readout}
\label{ssec:results_trainability}

Two properties decide whether the framework scales: training must keep a usable gradient as the model grows, and the trained model must be read out cheaply (\autoref{prop:rdm_readout}). Both rest on the same matchgate-compatible Givens rotations gates; here we test them numerically.

\paragraph{Trainability.} We measure the per-parameter gradient variance at random initialisation over a sweep of graph and embedding sizes (\autoref{fig:trainability}a). It declines only slowly with $N$ and follows the predicted dependence on the embedding dimension up to $58$ qubits, well above the exponential decay that would mark a barren plateau. Transferring trained parameters across sizes helps further: initialising each model from the TSP-5-trained model in a single step keeps the gradient signal one to two orders of magnitude above a random start at every size (\autoref{fig:trainability}b). The sweeps and breakdowns are in \suppref{si:trainability}.

\paragraph{Readout.} Deployment needs the embedding-register $1$-RDM read out cheaply, at the cost of \autoref{prop:rdm_readout}. The two computational-basis readouts recover only the diagonal of the $1$-RDM and miss its off-diagonal hopping entries, while the deterministic Hartree--Fock and matchgate-shadow readouts apply a matchgate basis change first and recover the full matrix. The difference shows up on the task: with exact features, the computational-basis $Z{+}ZZ$ readout reaches tour ratio $1.043$ on TSP-$5$, against $1.006$ for the full $1$-RDM that the two matchgate readouts estimate. The shot costs and task performance of all four readouts are in \suppref{si:measurement}, and the information the compression keeps in \suppref{si:rdm_jacobian}.

\section{Discussion}
\label{sec:discussion}

In this paper, we present a framework to adapt the structure and theoretical guarantees of GNNs from the classical literature to quantum circuits. We show how to represent and process graphs in a similar way, and how to obtain important theoretical properties that the classical literature identifies as the ones a graph neural network should have~\cite{gilmer2017neural, bronstein2021geometric}. We focused on three: message passing, in which each node gathers information from its neighbours and updates its state~\cite{gilmer2017neural, Velickovic2018}; permutation equivariance, in which relabelling the nodes relabels the output in the same way; and Weisfeiler--Leman expressivity, a standard measure of how finely a model distinguishes graphs~\cite{morris2019weisfeiler, maron2019invariant, chen2020can}.

We propose several methods, based on subspace-preserving gates, in order to prove that these properties are respected in a QGNN. The circuit performs message passing within the quantum state: the adjacency layer routes node amplitudes along the edges, the trainable evolution updates the per-node features, and re-uploading the graph at each layer makes the map nonlinear. The construction is exactly permutation equivariant at every value of the parameters, without tying or sharing weights, which is how earlier equivariant models impose the symmetry \cite{Skolik2023}. Our method can be adapted in new architectures, and in particular to other subspace preserving approaches \cite{Kerenidis2022, monbroussou_photonic_2025, bako2025fermionic}, in order to create useful quantum graph processing methods.

We validated these properties numerically on three datasets, including two with important real-world applications, and offer good results on large models (up to $56$ qubits). On the Cai--F\"urer--Immerman graphs, which are constructed so that no $1$-WL network can separate them, the model distinguished the pairs at exactly the level the ascent predicts~\cite{cai1992optimal}. On QM9, where the task is to predict the HOMO--LUMO gap, the error decreased as we raised the particle number at a fixed parameter budget~\cite{ramakrishnan2014quantum}. On the Euclidean travelling salesperson problem, the model produced near-optimal tours up to the largest instances we were able to simulate. These results, along with the theoretical motivations, illustrate how our framework can be used to provide good quantum methods that scale beyond small proofs of concept.

Finally, the approach is motivated by its scalability. Across every configuration we trained, the per-parameter gradient did not vanish as the subspace dimension grew. The barren-plateau results in the literature characterise broad families of variational circuits, typically generic or randomly initialised ans\"atze~\cite{larocca_barren_2025, ragone2024lie, Fontana2024, monbroussou2025trainability}; our model is not of that kind, as its graph information enters through a fixed initialisation and its trainable dynamics stay within a small structured register. The trainability we observe therefore sits alongside those results rather than against them: the general theory bounds what happens across a class, but it does not settle what a specific, structured instantiation does in practice, and the favourable behaviour we find is of that practical kind. The graph-style pre-training used here, where parameters fitted on small instances seed larger ones, is one way to preserve this behaviour as the graph grows, and it invites a theoretical treatment aimed at structured, pre-trained circuits rather than generic ones.

\section{Methods}
\label{sec:methods}

\subsection*{M1: Task-specific instantiations}
\setcounter{paragraph}{0}
{\color{black}The experiments use task-specific instantiations: TSP uses the layer stack of \autoref{ssec:results_quantum}; CFI and QM9 use Johnson mixing on the subset register. They also differ in the input encoding, the readout head, the loss, and any decoding step at inference.} The first, CFI graph discrimination, is the theoretical-validation experiment for the $j$-WL ascent of {\color{black}\autoref{thm:jwl_full_informal}} (\autoref{ssec:results_expressivity}); QM9 HOMO--LUMO gap regression (\autoref{ssec:results_qm9}) and Euclidean TSP edge prediction (\autoref{ssec:results_tsp}) are the two real-world tasks.

\paragraph{CFI graph discrimination at $j \geq 2$ (theoretical-validation experiment).} This discrimination test anchors {\color{black}\autoref{thm:jwl_full_informal}} on graph pairs that no $1$-WL-bounded network can distinguish. Inputs are the two graphs of the 18-vertex CFI($K_3$) pair, each expanded into $100$ random vertex relabellings; the relabelled instances are split at random $60/20/20$ into train, validation, and test sets. The target is the binary graph label. The encoder loads isomorphism-type subset features into the embedding register at $(D,k){=}(6,3)$. The per-subset embedding features are summed across the node register into a single permutation-invariant vector per graph, and a linear readout maps this vector to the class probability; the quantum layers stay at their initialisation and only the readout is fitted, for five seeds (\suppref{si:cfi}).

\paragraph{QM9 HOMO--LUMO gap regression at $j \in \{1,2,3\}$.} Inputs are atom-feature graphs from the full QM9 dataset ($130{,}831$ molecules); a molecule with more than $20$ atoms keeps its first $20$, which drops hydrogens from $21\%$ of the molecules; the target is the HOMO--LUMO gap in eV. Each $j$-subset of atoms carries its number of internal bonds, its sorted internal degrees (the full-graph degree at $j{=}1$) and histograms of its atom and bond types; at $j\ge 2$ subsets without an internal bond are dropped. An MLP maps these features to a $\binom{6}{3}{=}20$-dimensional embedding of each subset. Two rounds apply the equivariant mixing $\exp(i\alpha_r A_J)$ on the subset axis, re-upload the features and renormalise; the real and imaginary parts, summed over subsets after each round, pass through a small MLP head that produces the scalar prediction. Training minimises mean absolute error on a random $80/10/10$ split for $300$ epochs, with the target standardised on the training split; a complementary five-seed run on a $3{,}000$-molecule subset provides across-seed error bars. Per-$j$ settings and per-seed numbers are in \suppref{si:qm9}.

\paragraph{TSP edge prediction at $j{=}1$.} Inputs are the 2D Euclidean coordinates of the $N$ cities of a TSP instance, drawn uniformly from the unit square as in the construction of Vinyals et al.~\cite{Vinyals2015} and labelled with LKH tours; the training sets hold $5{,}000$, $10{,}000$, $200{,}000$, $10{,}000$ and $10{,}000$ instances at $N{=}5,10,20,30,50$, with $2{,}000$ validation and $2{,}000$ test instances at every $N$. The target is an edge-probability matrix $p_{ij} \in [0,1]$ over the complete graph on $N$ nodes. The encoder maps each coordinate pair through a linear layer with batch normalisation to the $\binom{6}{3}{=}20$ amplitudes of the embedding register at $(D,k){=}(6,3)$, loaded by the hierarchical loader of \autoref{ssec:results_quantum} at every layer. After the $L$ trainable iterations ($L{=}2$ at $N{=}5$ and $50$, $L{=}1$ otherwise) and the joint-register mixing layer, the per-node signed amplitudes pass through a learnable per-component affine map and a small MLP head on each city pair $(a,b)$, with input $(a, b, a{-}b, a\odot b)$, yielding the edge probabilities. Training minimises a binary cross-entropy with the positive class weighted by the ratio of negative to positive edges, on the training split, with early stopping on the validation loss, with optimiser settings in M2. The full evaluation protocol is in \suppref{si:tsp_benchmark}.

\subsection*{M2: Training and deployment protocol}
\setcounter{paragraph}{0}
We train the model in classical simulation and design the protocol so that the trained parameters can later be evaluated on quantum hardware without retraining. Each task is trained with the loss stated in its M1 instantiation. The TSP and QM9 models train all their parameters with a single Adam optimiser (for TSP with weight decay, cosine annealing of the learning rate and gradient-norm clipping at $1$); in the CFI experiment the quantum layers stay at their initialisation and only the two parameters of the readout are trained. The quantum angles start near the identity, drawn from a small zero-mean Gaussian (standard deviation $0.11$ to $0.3$; the two QM9 mixing angles are drawn once from a fixed seed), so that the gradient behaviour of \autoref{ssec:results_trainability} holds at initialisation.

Two aspects of the protocol concern scaling. Parameters fitted at one graph size serve as the initialisation at the next, since $(D,k,j)$ are independent of $N$ and the embedding-register evolution is invariant under $N$; the extra joint-mixer angles at the larger size are added in the canonical $S_N$-invariant ordering. This is the transfer initialisation used in the $N$-scaling sweeps of \autoref{ssec:results_trainability}. Because the trained circuit is matchgate-compatible, we envision carrying the same parameters onto quantum hardware at the deployment size, with the per-node features read out through the polynomial-shot $1$-RDM measurement of \autoref{ssec:results_trainability}. We do not perform this hardware step in the present work.

Reported TSP metrics are averaged over $10$, $10$, $5$, $5$ and $3$ seeds at $N{=}5,10,20,30,50$ and CFI metrics over five seeds, evaluated on a disjoint test split held out from training and model selection; TSP runs train for at most $300$, $500$, $800$, $1{,}000$ and $1{,}200$ epochs with early stopping (patience $50$), and the QM9 run for $300$ epochs on the full dataset. All models are simulated in PyTorch in the reduced basis: the unary node register at $j{=}1$ (TSP) and the weight-$j$ subspace of the node register at $j\ge 2$ (CFI, QM9). The mixing steps are checked against their PennyLane circuits to $10^{-10}$. In these simulations the encoders, the re-uploading, the renormalisation of each subset's embedding and the readout heads are classical operations.

\subsection*{M3: Resource estimates for quantum deployment}
\setcounter{paragraph}{0}
This paper validates the model in classical simulation and does not run it on hardware, so we give only a ballpark sense of what a future hardware run would need. We count gates for two of the trained TSP configurations at $(D,k){=}(6,3)$ and node-register particle number $j{=}1$ (Methods~M1): $N{=}20$ with one layer and $N{=}50$ with two. The two registers occupy $N+D$ qubits, from $26$ at $N{=}20$ to $56$ at $N{=}50$, well within the qubit count of present-day superconducting platforms such as Google Willow ($\sim 105$ qubits)~\cite{WillowNature2025} and IBM Heron ($\sim 156$ qubits).

A forward pass is built from particle-number-preserving Givens rotations, each compiling to a constant number of two-qubit gates~\cite{Anselmetti2021}, with the loader, which applies $19$ Givens per city controlled on that city's qubit, and the adjacency layer, with $\binom{N}{2}$ Givens per layer, dominating. One forward pass uses $974$ Givens at $N{=}20$ and $6{,}108$ at $N{=}50$, about $2.7\times10^3$ and $1.6\times10^4$ two-qubit gates before routing, with the per-block accounting in \suppref{si:resources}, \siref{tab:resource_blocks}. Reading the embedding-register $1$-RDM out to accuracy $\varepsilon$ takes $O(D^3/\varepsilon^2)$ shots per node, on the order of $10^5$ to a few times $10^6$ shots at $D{=}6$ and $\varepsilon{=}10^{-2}$, times the polynomial node-register projection overhead of \suppref{si:measurement}.

\subsection*{M4: Use of AI tools}
\setcounter{paragraph}{0}
The authors used large language model tools (Claude, Anthropic; Codex, OpenAI) to assist with software development and code review, checking of derivations, and editing of the manuscript text and its \LaTeX{} structure. The authors reviewed and verified all outputs and take full responsibility for the content of this work.

\section*{Data availability}
The Euclidean travelling-salesperson benchmark instances of \autoref{ssec:results_tsp} follow the random-uniform TSP construction of Vinyals et al.~\cite{Vinyals2015}; they are generated, with their LKH tours, by a script in the code repository with fixed seeds. The values for the equivariant quantum circuit of Skolik et al.~\cite{Skolik2023} are the validation results published with their code. The QM9 molecular regression dataset of \autoref{ssec:results_qm9} is loaded from the standard PyTorch Geometric distribution of QM9~\cite{ramakrishnan2014quantum}. The CFI graph families used for the theoretical-validation experiment of \autoref{ssec:results_expressivity} are generated programmatically from the construction of \cite{cai1992optimal}; the BREC benchmark pairs are taken from the official release of \cite{wang2024brec}. 

\section*{Code availability}
The code that supports the findings of this study is available at \url{https://github.com/SnehalRaj/mp-qgnns/}. The repository provides the CFI, QM9 and TSP models (reduced-basis PyTorch simulation, with the quantum mixing steps checked against PennyLane circuits), the training and evaluation protocols with the configuration of every reported run, the data-generation scripts, the classical baselines, the readout analysis of \suppref{si:measurement}, the BREC and set-$j$-WL computations, and the numerical checks quoted in the Supplementary Information. For each figure and table, the README gives the command that reproduces it, its run time, and whether it runs on a laptop or needs a GPU cluster.

\bibliography{references}

\section*{Funding}
BC, LM and EK acknowledge support from the EPSRC Quantum Advantage Pathfinder research (EP/X026167/1), and the Hub for Quantum Computing via Integrated and Interconnected Implementations (QCI3, EP/Z53318X/1) programs within the UK’s National Quantum Computing Center. S.R., A.J.F.-M. and R.M.S.F. declare no relevant funding.

\section*{Author contributions}
SR initiated and led the study, derived the theoretical results and performed the numerical simulations. BC and LM contributed to the architecture proposal, its scalability, and the interpretation of the numerical results. The manuscript was written by SR, LM, and BC. AJF-M. and RMSF contributed to the numerical simulations. EK supervised the project.

\section*{Competing interests}
The authors declare no competing interests.

\clearpage
\onecolumngrid
\appendix
\renewcommand{\appendixname}{Supplementary Note}
\renewcommand{\thesection}{\arabic{section}}
\renewcommand{\thesubsection}{\thesection.\arabic{subsection}}

\renewcommand{\figurename}{Supplementary Figure}
\renewcommand{\tablename}{Supplementary Table}
\renewcommand{\thetable}{\arabic{table}}
\renewcommand{\figureautorefname}{Supplementary Figure}
\renewcommand{\tableautorefname}{Supplementary Table}
\counterwithout{equation}{section}
\renewcommand{\theequation}{S\arabic{equation}}

\setcounter{figure}{0}
\setcounter{table}{0}
\setcounter{equation}{0}
\renewcommand{\theHfigure}{supp.\arabic{figure}}
\renewcommand{\theHtable}{supp.\arabic{table}}
\renewcommand{\theHequation}{supp.\arabic{equation}}

\begin{center}
{\Large\textbf{Supplementary Information}}
\end{center}
\bigskip
\noindent\hyperref[si:equivariance]{\ref*{si:equivariance}.~Proof of permutation equivariance (\texorpdfstring{\mainref{thm:equivariance_full}}{Theorem 1 of the main text})}\hfill\pageref*{si:equivariance}\par\smallskip
\noindent\hyperref[si:expressivity]{\ref*{si:expressivity}.~Proof of the \texorpdfstring{$j$}{j}-WL ascent (\texorpdfstring{\mainref{thm:jwl_full_informal}}{Theorem 2 of the main text})}\hfill\pageref*{si:expressivity}\par\smallskip
\noindent\hyperref[si:cfi]{\ref*{si:cfi}.~CFI graph discrimination: experimental details}\hfill\pageref*{si:cfi}\par\smallskip
\noindent\hyperref[si:qm9]{\ref*{si:qm9}.~QM9 HOMO--LUMO gap regression: experimental details}\hfill\pageref*{si:qm9}\par\smallskip
\noindent\hyperref[si:tsp_benchmark]{\ref*{si:tsp_benchmark}.~TSP benchmark full details}\hfill\pageref*{si:tsp_benchmark}\par\smallskip
\noindent\hyperref[si:trainability]{\ref*{si:trainability}.~Trainability derivations and DLA dimensions}\hfill\pageref*{si:trainability}\par\smallskip
\noindent\hyperref[si:measurement]{\ref*{si:measurement}.~Measurement strategy comparison}\hfill\pageref*{si:measurement}\par\smallskip
\noindent\hyperref[si:rdm_jacobian]{\ref*{si:rdm_jacobian}.~\RDM{} Jacobian rank derivation}\hfill\pageref*{si:rdm_jacobian}\par\smallskip
\noindent\hyperref[si:gate_ablation]{\ref*{si:gate_ablation}.~Gate-family ablation}\hfill\pageref*{si:gate_ablation}\par\smallskip
\noindent\hyperref[si:tsp_ablations]{\ref*{si:tsp_ablations}.~Frozen-parameter ablation}\hfill\pageref*{si:tsp_ablations}\par\smallskip
\noindent\hyperref[si:resources]{\ref*{si:resources}.~Resource estimates and gate decompositions}\hfill\pageref*{si:resources}\par\smallskip
\bigskip
\section{Proof of permutation equivariance (\texorpdfstring{\mainref{thm:equivariance_full}}{Theorem 1 of the main text})}\label{si:equivariance}

{\color{black}We prove the equivariance statement for the construction and readout specified below.}

\begin{maintheorem}[Exact $S_N$-equivariance]
{\color{black}
Let $\mathbf{x}=(\mathbf{X},A)$ contain the node features and adjacency matrix.
Consider the two-register QGNN with the uniform node state, subset-covariant
loader, sign-free Givens gates and shared readout specified in
\suppref{si:equivariance}. Let $\mathcal{K}$ record the gate
ordering and endpoint conventions defined there. Relabelling transports
these records with the graph.
For every parameter assignment $\boldsymbol{\theta}$ and every $\pi\in S_N$,
\begin{equation}
f_{\boldsymbol{\theta}}
  (\pi\cdot\mathbf{x};\pi\cdot\mathcal{K})
  =\pi\cdot f_{\boldsymbol{\theta}}(\mathbf{x};\mathcal{K}).
\end{equation}
Here $\pi$ permutes the node or subset outputs, and acts trivially on a
symmetrically pooled graph output. The trainable angles need no tying
within or across the three joint-mixing blocks.
}
\end{maintheorem}

\begin{proof}
{\color{black}
Let $\mathcal{K}$ record the persistent keys and ordered endpoints of the
adjacency and joint-mixing gates. Within each block, the keys determine a
fixed gate order. Under $\pi\in S_N$, each record retains its key, weight
and parameter assignment, while its node endpoints follow $\pi$.
We denote the resulting records by $\pi\cdot\mathcal{K}$.

Let $Q_\pi$ permute the node qubits and leave every embedding qubit unchanged.
On the node subspace of weight $j$, its restriction is
$P_\pi\otimes I_{\mathrm{emb}}$, where
$P_\pi\ket{T}=\ket{\pi(T)}$.
We use $Q_\pi$ on the full register space throughout the proof.
Its restriction to the joint subspace also covers states after
mixing.

We use the sign-free RBS gate $G_{ab}(\alpha)$ with ordered endpoints $(a,b)$.
On their local basis states, we take
\[
G_{ab}(\alpha)\ket{10}
 =\cos\alpha\ket{10}+\sin\alpha\ket{01},\qquad
G_{ab}(\alpha)\ket{01}
 =-\sin\alpha\ket{10}+\cos\alpha\ket{01}.
\]
It leaves $\ket{00}$ and $\ket{11}$ unchanged and has no parity string.
Extend $\pi$ to $\bar\pi$ by leaving the embedding indices fixed. Then
\begin{equation}
Q_\pi G_{ab}(\alpha)Q_\pi^\dagger
 =G_{\bar\pi(a)\bar\pi(b)}(\alpha).
\label{eq:gate_covariance}
\end{equation}
The ordered endpoints follow this map. Reversing them gives
$G_{ba}(\alpha)=G_{ab}(-\alpha)$.

\paragraph{Initial state.}
The node register starts in
$\ket{u_j}=\binom{N}{j}^{-1/2}\sum_{|T|=j}\ket{T}$.
The permutation only exchanges equally weighted terms, so
$P_\pi\ket{u_j}=\ket{u_j}$.
The embedding register starts in a fixed reference state
$\ket{\eta_0}$, on which $Q_\pi$ acts trivially.
Thus the initial density matrix $\rho_0$ satisfies
$Q_\pi\rho_0Q_\pi^\dagger=\rho_0$.

\paragraph{Hierarchical data loader $V(\mathbf{x})$.}
Write the controlled loader as
\[
V(\mathbf{x})=\sum_T\ket{T}\!\bra{T}\otimes V_T(\mathbf{x}),
\qquad
V_{\pi(T)}(\pi\cdot\mathbf{x})=V_T(\mathbf{x}).
\]
The sum runs over node basis states. Blocks outside the chosen node
subspace
may act as the identity.
At $j=1$, $V_T$ uses the corresponding node's features.
At $j>1$, it uses a symmetric function of the subset's features and induced
subgraph, with the same embedding circuit for corresponding subsets.
Relabelling the sum therefore gives
$V(\pi\cdot\mathbf{x})=Q_\pi V(\mathbf{x})Q_\pi^\dagger$.
This condition concerns the full loader unitary, since data reuploading
applies it to states produced by earlier layers.

\paragraph{Equivariant adjacency $\mathcal{A}$.}
Let $e_1,\ldots,e_m$ be the ordered node pairs in persistent-key order.
The angle $\alpha_r$ is proportional to the weight of $e_r$.
The keys preserve this sequence under relabelling, including when weights
are equal. The relabelled circuit is consequently
\begin{align}
\mathcal{A}(\pi\cdot A;\pi\cdot\mathcal{K})
 &=G_{\pi(e_m)}(\alpha_m)\cdots G_{\pi(e_1)}(\alpha_1)\notag\\
 &=Q_\pi\mathcal{A}(A;\mathcal{K})Q_\pi^\dagger.
\label{eq:adjacency_covariance}
\end{align}
The second equality follows from Eq.~\ref{eq:gate_covariance} and the
cancellation of adjacent factors $Q_\pi^\dagger Q_\pi$.
It holds for any initial order. No commutation between successive gates
is required.

\paragraph{Trainable evolution $W(\boldsymbol{\theta})$.}
The feature transformation has the form
$W(\boldsymbol{\theta})=I_{\mathrm{node}}\otimes
\widetilde W(\boldsymbol{\theta})$.
It acts only on the embedding register, so
$W(\boldsymbol{\theta})Q_\pi=Q_\pi W(\boldsymbol{\theta})$.
Its internal trainable angles can be arbitrary.

\paragraph{Joint mixing layer $M(\mathbf{x},\mathcal{K};
\boldsymbol{\theta}_M)$.}
The mixer preserves the total weight $j+k$ across both registers.
Its three blocks use the same persistent-key rule as the adjacency block.
Each trainable angle remains attached to its gate key under relabelling.
\begin{enumerate}
\item \emph{Node-pair block.}
Order the $\binom{N}{2}$ node pairs by their persistent keys and assign a
free angle $\theta^{(A)}_r$ to the pair at position $r$.
After relabelling, that position contains the transported ordered pair
with the same angle. Equation~\ref{eq:adjacency_covariance} applies
to this product for every choice of these angles.
\item \emph{Embedding-pair block.}
Apply Givens rotations to the $\binom{D}{2}$ embedding pairs with free
angles. Their endpoints remain unchanged, so this block commutes with
$Q_\pi$.
\item \emph{Cross-register block.}
Order the $ND$ node--embedding pairs by their persistent keys and assign a
free angle $\theta^{(C)}_r$ to the pair at position $r$.
A pair $(a,b)$ becomes $(\pi(a),b)$ with its key and angle unchanged.
Equation~\ref{eq:gate_covariance} therefore gives the same product
identity for this block.
\end{enumerate}
The fixed composition of these blocks satisfies
\begin{equation}
M(\pi\cdot\mathbf{x},\pi\cdot\mathcal{K};\boldsymbol{\theta}_M)
 =Q_\pi M(\mathbf{x},\mathcal{K};\boldsymbol{\theta}_M)Q_\pi^\dagger.
\label{eq:mixer_covariance}
\end{equation}
This identity restricts to the joint subspace of dimension
$\binom{N+D}{j+k}$ because $Q_\pi$ preserves total weight.
The parameter count remains $\binom{N}{2}+\binom{D}{2}+ND$, with no two
angles forced equal.

\paragraph{Composition and readout.}
The same cancellation applies to the complete circuit
$U_{\boldsymbol{\theta}}=M(W\mathcal{A}V)^L$.
Together with the invariant initial state, it gives
\[
\rho_{\boldsymbol{\theta}}
 (\pi\cdot\mathbf{x};\pi\cdot\mathcal{K})
 =Q_\pi\rho_{\boldsymbol{\theta}}
 (\mathbf{x};\mathcal{K})Q_\pi^\dagger.
\]
For a node outcome $T$, define the unnormalised conditional embedding state
$\sigma_T=\bra{T}\rho_{\boldsymbol{\theta}}\ket{T}$ and its probability
$p_T=\operatorname{Tr}\sigma_T$.
In the relabelled state, $\sigma'_{\pi(T)}=\sigma_T$ and
$p'_{\pi(T)}=p_T$.
For $p_T>0$, the conditional states $\sigma_T/p_T$ therefore agree for
corresponding subsets. Any fixed embedding observables give the same
features for those subsets. For zero-probability outcomes, we assign the
same fixed feature vector, such as zero.

A shared row head thus permutes its outputs with $T$.
At $j=1$, a shared head on pairs of node features gives edge outputs with
the same relabelling. Summing the subset features gives an invariant
graph representation. These readouts establish the claimed equivariance.
}
\end{proof}

{\color{black}
The theorem permits any initial gate ordering, provided that
the ordering, endpoint orientations and parameter assignments follow the
graph under relabelling. A graph-derived schedule gives equivariance directly
on graph data when $\mathcal K(\pi\cdot\mathbf x)=\pi\cdot\mathcal K(\mathbf x)$.
Graph symmetries can prevent unique tie-breaking.
}

\paragraph{Equivariance ablation table.} For the empirical comparison reported in \mainref{ssec:results_equivariance}, both arms are trained for $200$ epochs at every training-set size, with no early stopping, so the equivariant and non-equivariant variants see the same training budget. We report the validation BCE at a fixed cap epoch ($20$, $60$ and $100$ epochs for $n_{\rm train}{=}500$, $1000$ and $2000$), mean $\pm$ s.d.\ over three seeds, for the equivariant model ($5{,}773$ parameters) and the non-equivariant arm, the same architecture with the canonical edge ordering of $\mathcal{A}$ replaced by a fixed ordering ($5{,}813$ parameters).

\begin{table}[h]
\centering\small
\caption{\textbf{Sample-efficiency gap from $S_N$-equivariance, per training-set size.} Validation BCE at the cap epoch for the equivariance ablation on TSP-$5$, three seeds per cell, lower is better. The equivariant arm is the model; the non-equivariant arm replaces the canonical edge ordering of $\mathcal{A}$ by a fixed ordering.}
\label{tab:si_equivariance_capped}
\begin{tabular}{@{}c c c c@{}}
\toprule
$n_{\rm train}$ & Equivariant BCE & Non-equivariant BCE & reduction \\
\midrule
$500$  & $0.364 \pm 0.004$ & $0.535 \pm 0.010$ & $31.9\%$ \\
$1000$ & $0.304 \pm 0.010$ & $0.451 \pm 0.053$ & $32.7\%$ \\
$2000$ & $0.263 \pm 0.008$ & $0.424 \pm 0.038$ & $38.1\%$ \\
\bottomrule
\end{tabular}
\end{table}

\section{Proof of the \texorpdfstring{$j$}{j}-WL ascent (\texorpdfstring{\mainref{thm:jwl_full_informal}}{Theorem 2 of the main text})}\label{si:expressivity}

We prove that for $2\le j\le 4$, the QGNN at node-register particle number $j$ implements {\color{black}stable} set-based $j$-Weisfeiler--Leman colour refinement at generic parameters{\color{black}, with the encoder and depth specified below}. The proof has the two-implication form indicated in the main text: the unconditional upper bound that any two $j$-subsets sharing the same {\color{black}stable} $j$-WL colour have equal QGNN row vectors, and the generic-parameter lower bound that any two $j$-subsets with different colour have distinct rows for every $\boldsymbol{\theta}$ outside a closed measure-zero set. We close with {\color{black}explicit constructive separations in the experimental model}, the boundary case $j=1$, and the scope of the result.

The set-based $j$-Weisfeiler--Leman test in the convention of~\cite{morris2019weisfeiler} colours $j$-subsets of a graph's vertex set. Let $G=(V,E)$ have $N$ vertices and write $\binom{V}{j}$ for the collection of $j$-subsets. Each subset $S$ is initialised with the isomorphism type of the induced subgraph,
\begin{equation}\label{eq:jwl_init}
  c^{(0)}(S) \;=\; \mathrm{iso\!\!\!-\!\!\!type}\bigl(G[S]\bigr),
\end{equation}
and refined at each step by hashing the colour at $S$ together with the multiset of colours over the one-swap neighbourhood
\begin{equation}\label{eq:one_swap_si}
  \mathcal{N}_j(S) \;=\; \bigl\{ S' \in \tbinom{V}{j} \,:\, |S\,\triangle\,S'| = 2 \bigr\},
\end{equation}
producing
\begin{equation}\label{eq:jwl_update_si}
  c^{(\ell+1)}(S) \;=\; \mathrm{HASH}\!\Bigl(c^{(\ell)}(S),\; \bigl\{\!\!\bigl\{ c^{(\ell)}(S') : S' \in \mathcal{N}_j(S) \bigr\}\!\!\bigr\}\Bigr),
\end{equation}
where $\{\!\!\{\cdot\}\!\!\}$ denotes a multiset. {\color{black}Write $c^{(\infty)}$ for its stable colours, computed jointly across graphs.}

The QGNN state on the factored sub-product $\mathcal{H}_{\rm node}^{(j)}\otimes\mathcal{H}_{\rm emb}^{(k)}$ is the matrix $\Psi\in{\color{black}\mathbb{C}^{\binom{N}{j}\times\binom{D}{k}}}$ of \mainref{ssec:results_quantum}, with {\color{black}subset feature vector} $\Psi_S$ indexed by the $j$-subset $S\subseteq[N]$ and the column index running over $k$-subsets of $[D]$. The proof is about which row labels become distinguishable across QGNN iterations, so we treat each row $\Psi_S$ as the QGNN-side colour of $S$. {\color{black}Here the adjacency is Johnson mixing, $M(\alpha)=\exp(i\alpha A_J)$, where $(A_J)_{S,S'}=1$ exactly when $|S\triangle S'|=2$. An iteration applies the compound layer $W$, the shared encoder defined below, and then $M(\alpha)$.} The {\color{black}generator has the following neighbourhood} structure, which is the QGNN-side counterpart of the $j$-WL one-swap rule.

\begin{silemma}[One-swap neighbourhood matching]\label{lem:neighbourhood}
{\color{black}The Johnson generator $A_J$, restricted to the particle-number-$j$ subspace, couples row $\Psi_S$ to row $\Psi_{S'}$ if and only if $|S\triangle S'|=2$. The rows coupled to $S$ are exactly $\mathcal N_j(S)$.}
\end{silemma}

\begin{proof}
{\color{black}A nonzero entry $(A_J)_{S,S'}$ replaces one element of $S$ with one outside $S$, which is the definition of $\mathcal N_j(S)$.}
\end{proof}

\begin{mainjwltheorem}[Set-based $j$-WL expressivity]
{\color{black}
Fix $2\leq j\leq4$ and simple, unweighted graphs of size $N\geq j$.
Use the construction in the proof below for $L$
iterations, with feature width $d=\binom Dk\geq2$. Its shared analytic
encoder uses exact real and imaginary features. The initialisation separates
induced-subgraph isomorphism types. Parameters are independent across
iterations and shared across subsets and graphs.
\begin{enumerate}[label={\rm (\roman*)},leftmargin=2.5em]
\item For every parameter assignment and $\ell\leq L$,
\[
  c^{(\infty)}(S_1)=c^{(\infty)}(S_2)
  \ \Longrightarrow\ \Psi^{(\ell)}_{S_1}=\Psi^{(\ell)}_{S_2}.
\]
\item Outside a closed measure-zero set of parameters, for every $\ell\leq L$,
\[
  c^{(\ell)}(S_1)\neq c^{(\ell)}(S_2)
  \ \Longrightarrow\ \Psi^{(\ell)}_{S_1}\neq\Psi^{(\ell)}_{S_2}.
\]
\item At sufficient depth, $L\geq2\binom{N}{j}$ being sufficient, generic
parameters give
\[
  \Psi^{(L)}_{S_1}=\Psi^{(L)}_{S_2}
  \ \Longleftrightarrow\ c^{(\infty)}(S_1)=c^{(\infty)}(S_2).
\]
With the pooling in~\eqref{eq:jwl_pool}, generic parameters also give
$f_{\boldsymbol\theta}(G)=f_{\boldsymbol\theta}(G')$ exactly when the two graphs
have equal stable set-based $j$-WL colour histograms.
\end{enumerate}
}
\end{mainjwltheorem}

\begin{proof}
{\color{black}
\textit{Construction.}
For this expressivity construction, let $m=\binom{N}{j}$ and
$d=\binom{D}{k}\geq2$. The rows of $\Psi\in\mathbb{C}^{m\times d}$ describe
subset features. We use Johnson mixing,
\begin{equation}\label{eq:johnson_mixing}
  M(\alpha)=e^{i\alpha A_J},\qquad
(A_J)_{S,S'}=\begin{cases}1,&S'\in\mathcal{N}_j(S),\\0,&\text{otherwise.}\end{cases}
\end{equation}
The generator therefore couples precisely the one-swap neighbours from
\eqref{eq:one_swap_si}.
Between the compound layer and mixing, we apply a shared classical encoder
\begin{equation}\label{eq:analytic_encoder}
  E_{w,b}(z)=\frac{(1,w^\top z-b)}{\sqrt{1+(w^\top z-b)^2}},
\end{equation}
padded with zeros to width $d$. Let $O_\ell\in\mathrm{SO}(D)$ and write
$\Lambda^k(O_\ell)$ for its compound matrix.
One iteration is
\begin{equation}\label{eq:expressivity_iteration}
  Z=\Psi^{(\ell)}\Lambda^k(O_\ell)^\top,\qquad
  Y_S=E_{w_\ell,b_\ell}(\Re Z_S,\Im Z_S),\qquad
  \Psi^{(\ell+1)}=M(\alpha_\ell)Y.
\end{equation}
The encoder uses exact real and imaginary features. Parameters are independent
across iterations and shared across subsets and graphs. We omit the common
state-normalisation factor $m^{-1/2}$ throughout this feature description.

We initialise each row from the induced subgraph's closed-walk counts:
\begin{equation}\label{eq:closed_walk_init}
  \Psi^{(0)}_S=E_{w_{\mathrm{in}},b_{\mathrm{in}}}\!\left(
  \operatorname{Tr}(A[S,S]^2),\operatorname{Tr}(A[S,S]^3),
  \operatorname{Tr}(A[S,S]^4)\right).
\end{equation}
For $2\leq j\leq4$, these triples distinguish all induced-subgraph isomorphism
types. Choose the fixed projection $w_{\mathrm{in}}$ to separate distinct
triples.
For graph outputs, use independent pooling parameters $(w_*,b_*)$ and set
\begin{equation}\label{eq:jwl_pool}
  f_{\boldsymbol\theta}(G)=\sum_{S\in\binom{V}{j}}
  [E_{w_*,b_*}(\Re\Psi^{(L)}_S,\Im\Psi^{(L)}_S)]_1,
\end{equation}
where $[\cdot]_1$ selects the first encoder coordinate.

For the upper bound, stable colours refine the initial isomorphism types.
Thus equal stable colours have equal initial features.
The stable partition is equitable: subsets in one class have equal numbers
of neighbours in each class. Consequently, $A_J$ preserves features that are
constant on each class. Its powers and $M(\alpha)$ preserve them as well.
The shared compound layer and encoder also preserve equality within classes.
Induction proves the upper bound, including comparisons across the disjoint
coloured Johnson graphs.

For the lower bound, the initialisation separates isomorphism types at
$j=2,3,4$, as direct enumeration verifies. The encoder preserves distinct
projections because its second-to-first coordinate ratio is $w^\top z-b$.
We also need it to separate sums over unequal feature multisets.
For a finite set of distinct features, choose $w$ so their projections $t_i$
are distinct. The first encoder coordinates are
$h_i(b)=(1+(b-t_i)^2)^{-1/2}$.
These functions are linearly independent: their large-$b$ series coefficients
are monic polynomials of successive degrees in $t_i$.
A Vandermonde argument therefore excludes a nonzero linear relation.
A common bias $b$ can consequently separate all finitely many unequal
multiset sums while retaining distinct individual features.

Assume the lower bound holds at iteration $\ell$.
The compound matrix is invertible, so it preserves unequal features.
Choose the encoder as above and consider two subsets with different
$(\ell+1)$-step colours.
If their current features differ, their encoded rows differ at $\alpha=0$,
where $M(0)=I$.
Otherwise, their neighbour-colour multisets differ.
The inductive hypothesis gives unequal feature multisets, whose encoded sums
differ. Hence
\begin{equation}\label{eq:jwl_derivative_witness}
  \left.\frac{\mathrm{d}}{\mathrm{d}\alpha}
  \bigl(\Psi^{(\ell+1)}_{S_1}-\Psi^{(\ell+1)}_{S_2}\bigr)\right|_{\alpha=0}
  =i\left[\sum_{S'\in\mathcal{N}_j(S_1)}Y_{S'}
          -\sum_{S'\in\mathcal{N}_j(S_2)}Y_{S'}\right]\neq0.
\end{equation}
In either case, the squared feature difference is a nonzero real-analytic
function of the full parameter vector. Its zero set is closed and has
measure zero. There are finitely many required pairs and iterations at fixed
$N$ and $L$. Their exceptional sets have a closed measure-zero union, which
proves the generic lower bound.

Refinement on a pair of graphs stabilises within $2\binom{N}{j}$ rounds.
At that depth, the upper and lower bounds identify the same subset classes.
Equal stable colour histograms give equal pooled outputs by the upper bound.
For unequal histograms, the same encoder argument separates the sums in
\eqref{eq:jwl_pool} at generic pooling parameters.
This proves the graph-level statement.
}
\end{proof}

{\color{black}A constructive separation in the experimental model of \suppref{si:cfi} uses the unnormalised subset encoding
\begin{equation}\label{eq:generic_init}
 (1,\max\{0,e-(\tbinom j2-1)\},0,\ldots),
\end{equation}
where $e$ counts induced edges. The existing ReLU encoder realises this map; its initial sum pool counts $j$-cliques up to a common factor. At $j=3$, it separates $C_6$ from two triangles (zero versus two cliques); at $j=4$, it separates the rook and Shrikhande graphs (eight versus zero). At the preceding orders, both pairs have matching equitable induced-type partitions, so (i) precludes separation. The experimental model satisfies (i); the generic lower bound uses the analytic encoder above.}

At $j=1$ the rows of $\Psi$ are indexed by single vertices, and the closed-walk initialisation~\eqref{eq:closed_walk_init} reduces to the trivial fingerprint $(\Tr(A[\{v\},\{v\}]^m))_{m=2,3,4}=(0,0,0)$ for every singleton, so all initial rows agree. {\color{black}Johnson mixing preserves these equal features. \mainref{thm:jwl_full_informal} does not apply at $j=1$; the degree features used in the experimental model require a separate upper-bound argument.}

The scope of the result is bounded by three constraints. The closed-walk fingerprint~\eqref{eq:closed_walk_init} fails at $j\ge 5$, where collisions among non-isomorphic five-vertex graphs prevent injective initialisation; extending the theorem requires a richer subgraph invariant such as {\color{black}an exact induced-subgraph isomorphism type}. The proof assumes {\color{black}Johnson mixing}, so restricting to nearest-neighbour or other structured connectivities would change the reachable neighbourhoods and replace \autoref{lem:neighbourhood} with a weaker statement. Finally, the proof is for the global set-based $j$-WL of~\cite{morris2019weisfeiler}, not the local $\delta$-$k$-LWL variant of Morris--Rattan--Mutzel~\cite{morris2020weisfeiler}, which uses a graph-dependent neighbourhood at each step.

\section{CFI graph discrimination: experimental details}\label{si:cfi}

The $j$-Weisfeiler--Leman ascent of \mainref{thm:jwl_full_informal} is tested on the CFI($K_3$) pair and on the BREC benchmark. This note documents the QGNN-side instantiation at $j\ge 2$, exact set-based $j$-WL discrimination across the $219$ BREC pairs with at most $25$ vertices, and the classical $1$-WL and $3$-WL baselines trained end-to-end on the same graphs at parameter budgets that bracket the QGNN.

At $j\ge 2$ the experiment simulates the node register in its particle-number-$j$ subspace, one amplitude vector per $j$-subset $S$, in PyTorch at $(D,k){=}(6,3)$. A small MLP maps the isomorphism-type features of the induced subgraph $G[S]$ (its number of edges and its sorted degree sequence) to a unit-norm $\binom{6}{3}$-dimensional embedding of $S$. Each of three rounds then applies the equivariant mixing $\exp(i\alpha_r A_J)$, with $A_J$ the adjacency of the Johnson graph on $j$-subsets and $\alpha_r$ a scalar, adds a re-uploaded encoding of the same features, and renormalises every subset's embedding. The mixing step equals the evolution under the hopping Hamiltonian restricted to particle number $j$, checked against its PennyLane circuit to $10^{-12}$; the encoders, the re-uploading sum and the renormalisation are classical. A sum over subsets of the real and imaginary parts of the embedding after every round gives one permutation-invariant vector per graph, and a Fisher readout maps it to a class probability.

The dataset is the Cai--F\"urer--Immerman pair CFI($K_3$) on $N{=}18$ vertices, generated from the construction of~\cite{cai1992optimal}; exact set-based $j$-WL refinement first separates the pair at $j{=}3$ (at its second round; $1$-WL never separates it, and folklore $2$-WL does at its third round). The two graphs are $2$-regular: one is two disjoint $9$-cycles, the other an $18$-cycle. The BREC benchmark of~\cite{wang2024brec} holds $400$ pairs of $1$-WL-indistinguishable graphs in six groups; we test the $219$ pairs with at most $25$ vertices from its official pair file: $60$ Basic, $58$ Regular, $100$ Extension and $1$ CFI pair. We use BREC to record what exact set-based $j$-WL discriminates at each level before reporting the learned QGNN results.

\begin{table}[h]
\centering\small
\caption{\textbf{Exact set-based $j$-WL discrimination across BREC groups.} Pairs separated out of the BREC pairs with at most $25$ vertices (official pair file of~\cite{wang2024brec}), by group. Here $j$ is the set-based $j$-WL of~\cite{morris2019weisfeiler}, the level of node-register particle number $j$ in the model; refinement runs to a stable colouring.}
\label{tab:supp_brec_groups}
\begin{tabular}{@{}c c c c c c@{}}
\toprule
$j$ & Basic & Regular & Extension & CFI & Total \\
\midrule
pairs & $60$ & $58$ & $100$ & $1$ & $219$ \\
$1$ & $0$ & $0$ & $0$ & $0$ & $0$ \\
$2$ & $0$ & $0$ & $0$ & $0$ & $0$ \\
$3$ & $60$ & $50$ & $100$ & $1$ & $211$ \\
$4$ & $60$ & $58$ & $100$ & $1$ & $219$ \\
\bottomrule
\end{tabular}
\end{table}

The deterministic test runs the colour refinement defined in~\eqref{eq:jwl_update_si} of \suppref{si:expressivity} until the colouring is stable: at each round, every $j$-subset is recoloured by hashing its current colour together with the multiset of colours over its one-swap neighbourhood, with the isomorphism-type initialisation. No model is trained, and no random seed enters the test; the only inputs are the two graphs of a pair. Every BREC pair is $1$-WL-indistinguishable by construction, and no pair separates at $j{\le}2$. At $j{=}3$, $211$ of the $219$ pairs separate: all Basic, Extension and CFI pairs and $50$ of the $58$ Regular pairs. At $j{=}4$ all $219$ separate. Set-based $j$-WL upper-bounds what the QGNN at particle number $j$ can discriminate; the main text shows this bound is attained on CFI($K_3$) at the predicted threshold (\mainref*{ssec:results_expressivity} and \mainref{fig:cfi_discrimination}).

The classical baselines probe the same pair with $1$-WL and $3$-WL networks trained end-to-end at parameter budgets that bracket the QGNN. We use two encoders with distinct expressivity classes. The Graph Isomorphism Network (GIN) of Xu et al.~\cite{Xu2019} is the canonical $1$-WL representative: each layer applies the aggregation $h_v^{(\ell+1)} = \mathrm{MLP}^{(\ell)}\!\bigl((1+\epsilon)h_v^{(\ell)} + \sum_{u\in\mathcal{N}(v)} h_u^{(\ell)}\bigr)$ with a learnable scalar $\epsilon$ and a $2$-layer MLP per layer. The Provably Powerful Graph Network (PPGN) of Maron et al.~\cite{maron2019provably} represents graphs as $N\times N$ tensors and updates them by blocks of the form $\mathrm{MLP}_3\bigl(X, M_1 M_2\bigr)$, where $X$ is the block input and $M_1, M_2$ are $\mathrm{MLP}_1$- and $\mathrm{MLP}_2$-transformed copies of it multiplied as matrices channel by channel; the trained baselines take the adjacency matrix as their single input channel. PPGN is $3$-WL class, equivalent to the $2$-folklore-WL test.

Both baselines and the QGNN use the same dataset: the two graphs of the pair, each with $100$ random vertex relabellings, split $60$/$20$/$20$ into train, validation and test sets. The baselines are trained end-to-end: each encoder feeds into a sum-pool over node features and a $2$-layer MLP head to a single binary logit, trained with Adam at learning rate $10^{-3}$, batch size $32$, binary cross-entropy on the graph label, and early stopping on validation accuracy with patience $30$ over $200$ epochs, for seeds $0$ to $4$. The QGNN keeps its equivariant quantum layers fixed at initialisation; a Fisher readout, a whitening followed by a trained one-dimensional affine map, is fitted on its sum-pooled embeddings over $300$ epochs at learning rate $5\times10^{-3}$ in minibatches of $32$, keeping the checkpoint with the best validation accuracy, for seeds $0$ to $4$ (Methods~M1). The split is not stratified, so below the separating level, where the readout returns a constant, the test accuracy is the share of one class in each seed's test split. \mainref{fig:cfi_discrimination} reports the mean and standard deviation of test accuracy across the seeds.

The QGNN has $3{,}029$, $3{,}157$ and $3{,}285$ parameters (model plus readout) at $j{=}1$, $2$ and $3$, of which only the $2$ readout parameters are trained, and reaches test accuracy $0.47 \pm 0.08$, $0.47 \pm 0.08$ and $1.00 \pm 0.00$; its output spread under relabelling at $j{=}3$ is of order $10^{-3}$, the floating-point level of the round-two signal. The baseline budgets bracket these counts. GIN with three layers at $1{,}684$, $3{,}676$ and $6{,}436$ parameters stays at chance at every budget, $0.45 \pm 0.07$ over the fifteen runs, as the $1$-WL identity $\mathrm{GIN}\equiv 1\text{-WL}$ of~\cite{Xu2019} and the construction of~\cite{cai1992optimal,morris2019weisfeiler} require. PPGN with one or two blocks at $1{,}329$, $2{,}209$ and $4{,}097$ parameters also stays at chance, $0.49 \pm 0.08$ over fifteen runs. This is a property of the architecture on this pair rather than a training failure: in a separate check with the standard input (adjacency and identity channels) and readout (diagonal and off-diagonal sums), the pooled PPGN output on the two graphs of the pair is identical to within $10^{-15}$ for all $81$ random networks tested (widths $16$, $64$, $128$; depths $1$, $4$, $16$; ReLU, tanh and GELU; three seeds), so gradient training receives no signal, while $54$ of the same $81$ networks separate $C_6$ from two triangles, a $1$-WL-equivalent pair that folklore $2$-WL separates. Exact folklore $2$-WL refinement does separate the pair at its third round, so the $j{=}3$ level is confirmed deterministically; the bilinear PPGN block at generic weights does not realise that refinement on this pair.

The code repository reproduces every run in this note: the QGNN climb, both baselines, the BREC table and the exact refinements, with per-seed results.

\section{QM9 HOMO--LUMO gap regression: experimental details}\label{si:qm9}

The $j$-WL ascent on QM9 reported in \mainref{ssec:results_qm9} (\autoref{tab:qm9_jsweep}) is run on the full QM9 dataset of $130{,}831$ molecules~\cite{ramakrishnan2014quantum}, loaded from the standard PyTorch Geometric distribution. A molecule with more than $20$ atoms keeps its first $20$; since the hydrogens are listed last, $21\%$ of the molecules lose between one and nine of them. The target is the HOMO--LUMO gap (target index $4$), in eV after the customary Hartree-to-eV conversion.

\paragraph{Architecture and training.} All three runs use the same $(D,k){=}(6,3)$ embedding, two mixing rounds $\exp(i\alpha_r A_J)$, each followed by re-uploading and renormalisation, sum-pooling readout across the node register, and a small MLP head producing the scalar prediction. The encoder uses isomorphism-type atom features (full-graph degree at $j{=}1$, edge-gating at $j \geq 2$, where subsets without an internal bond are dropped). Only the node-register particle number $j$ varies across runs. Each model is trained for $300$ epochs on a random $80/10/10$ split (seed of the run) with one Adam optimiser (learning rate $10^{-3}$, weight decay $10^{-4}$, cosine annealing to $10^{-5}$, batch $64$); the target is standardised on the training split, and the best validation state is restored before testing. We report the test MAE at each $j$.

\paragraph{Results, parameter counts, and the classical baseline.}
\begin{table}[h]
\centering\small
\caption{\textbf{QM9 HOMO--LUMO gap: mean absolute error and trainable parameter counts on the full dataset.} Single seed (42), $300$ epochs. All three entries are test MAE.}
\label{tab:qm9_jsweep}
\begin{tabular}{@{}l r c@{}}
\toprule
Model & Params & MAE (eV) \\
\midrule
QGNN, $j{=}1$ & $7{,}039$ & $0.460$ \\
QGNN, $j{=}2$ & $7{,}135$ & $0.282$ \\
QGNN, $j{=}3$ & $7{,}231$ & $0.236$ \\
\midrule
Morris $1$-GNN~\cite{morris2019weisfeiler} & $849{,}569$ & $0.121$ \\
\bottomrule
\end{tabular}
\end{table}

The error falls monotonically with $j$. The classical $1$-GNN of Morris et al.~\cite{morris2019weisfeiler}, trained in its own pipeline (its own split, no atom cap, all atom features, 3D interatomic distances, a mean-squared-error loss, one seed), reaches a lower error of $0.121$~eV at $849{,}569$ parameters; the QGNN at $j{=}3$ sits at roughly $1.9\times$ that error with about $120\times$ fewer parameters. The parameter count grows by $96$ per step in $j$ because the subset features carry the $j$ sorted internal degrees, which widens the input of the three encoders; the mixer has two angles at every $j$. The embedding capacity $(D,k){=}(6,3)$ is held fixed across all three runs, so the error reduction reflects the added node-register resolution and not added feature capacity.

\paragraph{Robustness across seeds.} A complementary five-seed run on a $3{,}000$-molecule subset (seeds $42$--$46$, molecules capped at $18$ atoms, which drops hydrogens from $45\%$ of them, $100$ epochs) reproduces the same monotone descent across $j$: test MAE $0.611 \pm 0.025$~eV at $j{=}1$, $0.469 \pm 0.030$~eV at $j{=}2$ and $0.448 \pm 0.031$~eV at $j{=}3$. The absolute values are higher than on the full set because the subset is data-starved; the ordering in $j$ holds across seeds.
\section{TSP benchmark full details}\label{si:tsp_benchmark}

This note documents the full protocol behind the multi-seed Euclidean TSP edge-prediction numbers cited in \mainref*{ssec:results_tsp} and \mainref{fig:tsp_demonstration}. The numbers reported in the main text (validation BCE and test tour ratio at $N \in \{5, 10, 20, 30, 50\}$) are produced by the QGNN at $(D, k) = (6, 3)$ on random-uniform instances of the kind used in the equivariant quantum circuit benchmark of Skolik et al.~\cite{Skolik2023}.

\paragraph{Dataset.}
The benchmark instances follow the random-uniform TSP construction of Vinyals et al.~\cite{Vinyals2015}, adopted in the EQC benchmark of Skolik et al.~\cite{Skolik2023}: $N$ city coordinates sampled uniformly in the unit square, with optimal Hamiltonian cycles computed by the Lin--Kernighan--Helsgaun heuristic (we use the \texttt{elkai} Python binding). The train, validation and test splits are drawn with random seeds $42$, $1042$ and $2042$ by the generation script of the code repository. Training uses $5{,}000$, $10{,}000$, $200{,}000$, $10{,}000$ and $10{,}000$ instances at $N = 5, 10, 20, 30, 50$; validation and test use $2{,}000$ instances each at every $N$. The same splits are reused across seeds, so the seed-to-seed variation reflects optimisation stochasticity rather than data resampling.

\paragraph{Per-size hyperparameters.}
The embedding setting is $(D, k) = (6, 3)$ at every $N$, giving $\binom{6}{3} = 20$-dimensional particle-number subspaces and the same fixed embedding throughout. Only the depth of the trainable evolution and the width of the readout head scale with $N$. \autoref{tab:si_tsp_hparams} lists the per-size values; they are the configurations of the cluster jobs that produced the numbers below, held in \texttt{configs.py} of the code repository.

\begin{table}[h]
\centering
\caption{\textbf{QGNN hyperparameters per TSP size.} The embedding $(D, k) = (6, 3)$ is fixed across all sizes; depth $L$, pyramids per layer $P$, hidden dimension, training length, and learning rate are tuned per $N$. The trainable evolution $W$ uses nearest-neighbour pyramid connectivity throughout.}
\label{tab:si_tsp_hparams}
\footnotesize
\begin{tabular}{@{}lccccc@{}}
\toprule
 & TSP-5 & TSP-10 & TSP-20 & TSP-30 & TSP-50 \\
\midrule
$D$ & 6 & 6 & 6 & 6 & 6 \\
$k$ & 3 & 3 & 3 & 3 & 3 \\
Layers $L$ & 2 & 1 & 1 & 1 & 2 \\
Pyramids per layer $P$ & 2 & 1 & 4 & 4 & 4 \\
Readout hidden dimension & 64 & 128 & 128 & 128 & 256 \\
Epochs (max) & 300 & 500 & 800 & 1000 & 1200 \\
Learning rate & $5 \times 10^{-3}$ & $9.5 \times 10^{-4}$ & $2.3 \times 10^{-3}$ & $2.3 \times 10^{-3}$ & $4.2 \times 10^{-3}$ \\
Weight decay & $0$ & $3.8 \times 10^{-5}$ & $4.6 \times 10^{-5}$ & $4.6 \times 10^{-5}$ & $8.6 \times 10^{-5}$ \\
Batch size & 64 & 64 & 64 & 256 & 128 \\
Beam width & 2000 & 2000 & 2000 & 2000 & 2000 \\
Trainable parameters & 5{,}813 & 11{,}149 & 11{,}399 & 11{,}704 & 23{,}486 \\
\bottomrule
\end{tabular}
\end{table}

\paragraph{Optimisation.}
One Adam optimiser updates all parameters, classical and quantum, at the per-size learning rate and weight decay of \autoref{tab:si_tsp_hparams}. Cosine annealing brings the learning rate to $10^{-5}$ over the maximum epoch budget. Gradients are clipped to unit norm before each step. Early stopping triggers after $50$ epochs without improvement on the validation BCE; the lowest-validation-BCE checkpoint is restored before test-set evaluation. The loss weights the positive (tour) edges by the ratio of negative to positive edges.

\paragraph{Multi-seed protocol.}
At $N = 5$ and $N = 10$ we run ten independent seeds (random states $42, 43, \ldots, 51$) on the same train/val/test split, each going through the full training and early-stopping pipeline. At $N = 20$, $N = 30$, and $N = 50$ we run five, five, and three seeds respectively (random states $42$ to $46$ at $N = 20, 30$ and $43$ to $45$ at $N = 50$). Three of these runs ($N = 20$ seeds $44$ and $45$, $N = 50$ seed $44$) reached the cluster time limit and one ($N = 50$ seed $45$) was stopped at epoch $170$ before early stopping triggered; their numbers come from the best checkpoint reached. Reported test metrics are means and standard deviations over seeds.

\paragraph{Decoding.}
The trained model outputs an edge-probability matrix $p_{ij} \in [0,1]$ over the complete graph on $N$ nodes. We decode these probabilities by beam search, following the edge-heatmap approach of Joshi et al.~\cite{joshi2019efficient}. Each edge receives score $s_{ij}=\log\max(p_{ij},10^{-10})$. From each start city $c=\operatorname{round}(iN/5)$, $i=0,\ldots,4$, we retain up to $2{,}000$ partial paths ranked by cumulative score. Each extension visits an unused city. After all cities have been visited, we add the closing-edge score and select the highest-scoring tour across all retained paths and starts. Tour selection uses model scores only. We select the settings on validation and keep them fixed for the test evaluation at every size. We evaluate graphs in stored order, in batches of $128$, because BatchNorm uses batch statistics. The tour ratio is the mean over the test set of the decoded tour length divided by the reference tour length.

\paragraph{Results.}
\autoref{tab:si_tsp_results} reports the multi-seed metrics: the validation BCE at the best checkpoint and the test-set tour ratio. The TSP-5 and TSP-10 numbers are ten-seed runs; the TSP-20, TSP-30, and TSP-50 numbers are five-, five-, and three-seed cluster runs. Indeed, the validation BCE at $N = 20$ has the same magnitude as at $N = 10$ even though the underlying graph is twice as large, consistent with the fixed embedding $(D, k) = (6, 3)$ controlling the readout cost across sizes.

\begin{table}[h]
\centering
\caption{\textbf{QGNN metrics on random-uniform TSP at $(D, k) = (6, 3)$.} Means and standard deviations are over independent seeds; the BCE is the validation BCE at the best checkpoint, and the tour ratio is measured on a held-out test split disjoint from training and validation with beam width $2{,}000$ and five start cities. The lower panel reproduces the published means of Teo et al.~\cite{Teo2025}.}
\label{tab:si_tsp_results}
\footnotesize
\begin{tabular}{@{}lcccc@{}}
\toprule
$N$ & Trainable params & Validation BCE & Tour ratio & Seeds \\
\midrule
5  & 5{,}813  & $0.226 \pm 0.004$ & $1.0040 \pm 0.0005$ & 10 \\
10 & 11{,}149 & $0.374 \pm 0.003$ & $1.0201 \pm 0.0007$ & 10 \\
20 & 11{,}399 & $0.354 \pm 0.002$ & $1.0330 \pm 0.0016$ & 5 \\
30 & 11{,}704 & $0.325 \pm 0.006$ & $1.0523 \pm 0.0014$ & 5 \\
50 & 23{,}486 & $0.334 \pm 0.013$ & $1.1405 \pm 0.0168$ & 3 \\
\bottomrule
\end{tabular}

\par\medskip
\begin{tabular}{@{}lccccc@{}}
\toprule
Published comparison & TSP-5 & TSP-10 & TSP-20 & TSP-30 & TSP-50 \\
\midrule
QGNN (this work) & 1.0040 & 1.0201 & 1.0330 & 1.0523 & 1.1405 \\
Teo et al., RL & 1.032 & 1.057 & 1.105 & 1.135 & 1.160 \\
Teo et al., SIGS & 1.018 & 1.046 & 1.095 & 1.141 & 1.161 \\
\bottomrule
\end{tabular}
\end{table}

The $N{=}30$ and $N{=}50$ rows are multi-seed cluster runs (five and three seeds). The tabulated tour ratios are test-set means and standard deviations at the best validation checkpoint; the BCE is the validation value at that checkpoint.

\paragraph{Published quantum baselines.}
The comparison in \autoref{tab:si_tsp_results} uses the mean tour ratios reported by Teo et al.~\cite{Teo2025} for their depth-one equivariant circuit. Their RL and size-invariant grid search (SIGS) policies construct tours greedily from one start city. Our test split contains $2{,}000$ uniform-square instances with elkai/LKH reference tours. Teo et al. use $1{,}000$ separately generated uniform-square instances at each size. They use Held--Karp reference tours below $15$ cities and the best of $20$ LKH-3 runs at larger sizes. These results compare the reported pipelines under their respective protocols.

\paragraph{Classical simulation cost.}
All results above are produced by exact state-vector simulation confined to the particle-number subspace, of dimension $\binom{N+D}{j+k}$. At the embedding half-filling $k=D/2$ used throughout (the largest subspace for a given $D$), this is $\binom{N+D}{4}$ at $j{=}1$, growing from $330$ at $N{=}5$ to $367{,}290$ at $N{=}50$ (\autoref{tab:si_simcost}). One forward-and-backward pass scales as roughly $D_{\rm sub}^{1.1}$ across $N$ at fixed $(D,k)$ and $D_{\rm sub}^{1.3}$ across $(D,k)$ at fixed $N$, so the cost, while polynomial in $N$ at fixed $(D,k)$, is a steep polynomial: a single training run at $N\ge 20$ took two to four days on one GPU. Reaching $N{=}50$ ($56$ qubits) is set by training walltime; extrapolated to $N{\approx}100$ a single pass takes tens of minutes and a full run takes months, beyond what we can simulate classically.

\begin{table}[h]
\centering
\caption{\textbf{Classical simulation cost of the TSP runs.} All runs use $(D,k){=}(6,3)$ and $j{=}1$. The simulated state lives in the particle-number subspace of dimension $\binom{N+D}{4}$. The forward-and-backward time is measured on one CPU core for the model of the gradient-variance study of \suppref{si:trainability} at the same $(N,D,k)$. The last column is the full-budget training time, the measured time per epoch multiplied by the maximum epoch budget, on one A100 (a slice of one at $N{=}10$ and $30$); the runs of \autoref{tab:si_tsp_results} stopped earlier, or at the cluster time limit of $72$ hours. The hours depend on the train set size as well as $N$, so they do not grow with $N$ in order: $N{=}20$ takes more hours than $N{=}30$ because it trains on more instances.}
\label{tab:si_simcost}
\footnotesize
\begin{tabular}{@{}lccccc@{}}
\toprule
$N$ & Qubits $N+D$ & Subspace dim $\binom{N+D}{4}$ & Fwd+bwd (s) & Train set & Full run (GPU-h) \\
\midrule
5  & 11 & 330      & 0.07 & 5{,}000   & $<1$ \\
10 & 16 & 1{,}820  & 0.12 & 10{,}000  & ${\sim}4$ \\
20 & 26 & 14{,}950 & 1.0  & 200{,}000 & ${\sim}316$ \\
30 & 36 & 58{,}905 & 7.1  & 10{,}000 & ${\sim}47$ \\
50 & 56 & 367{,}290 & 114 & 10{,}000  & ${\sim}356$ \\
\bottomrule
\end{tabular}
\end{table}

Raising the embedding makes the simulation far more expensive. \autoref{tab:si_simcost_dk} compares the three half-filling embeddings $(6,3)$, $(8,4)$, and $(10,5)$ at $N{=}20$: a single forward-and-backward pass grows from $1$ second to over $100$. At larger $N$ the gradient-variance study reached $N{=}50$ with $(8,4)$ but only $N{=}25$ with $(10,5)$ before reaching the walltime and memory limits.

\begin{table}[h]
\centering
\caption{\textbf{Simulation cost at $N{=}20$ across embedding sizes.} All at half-filling $k{=}D/2$, $j{=}1$, for the model of the gradient-variance study. Forward-and-backward time is one CPU core.}
\label{tab:si_simcost_dk}
\footnotesize
\begin{tabular}{@{}lccc@{}}
\toprule
$(D,k)$ & Qubits $N+D$ & Subspace dim $\binom{N+D}{j+k}$ & Fwd+bwd (s) \\
\midrule
$(6,3)$  & 26 & 14{,}950  & 1.0 \\
$(8,4)$  & 28 & 98{,}280  & 8.4 \\
$(10,5)$ & 30 & 593{,}775 & 116 \\
\bottomrule
\end{tabular}
\end{table}
\section{Trainability derivations and DLA dimensions}\label{si:trainability}

The two-register QGNN trains by gradient descent on a balanced binary cross-entropy cost. For optimisation to remain feasible as the graph size $N$ and the embedding-register size $D$ grow, the variance of every single-parameter gradient must stay above a polynomial floor. The barren-plateau phenomenon~\cite{mcclean2018barren, Cerezo2021}, in which this variance decays as $4^{-(N+D)}$ for unstructured deep circuits, would close off training at moderate qubit count. We give a closed-form lower bound on the gradient variance for the full pipeline of \mainref{ssec:results_quantum}, valid for any Hermitian readout used in the architecture, with the polynomial-not-exponential scaling made explicit. The bound takes the form of a single proposition (\autoref{prop:trainability_si}); we derive it by combining a single-register Lie-algebraic argument with a joint-register mixing argument.

We work with $N$ node-register qubits at fixed particle number $j$ and $D$ embedding-register qubits at fixed particle number $k$. The relevant subspace dimensions are
\begin{equation}\label{eq:dims}
  d_N \;=\; \binom{N}{j}, \qquad d_E \;=\; \binom{D}{k}, \qquad d_J \;=\; \binom{N+D}{j+k}.
\end{equation}
The factored sub-product $\mathcal{H}_{\mathrm{node}}^{(j)}\otimes\mathcal{H}_{\mathrm{emb}}^{(k)}$ has dimension $d_N d_E$; the joint-register mixing layer $M(\boldsymbol{\theta}_M)$ preserves only the total weight $j+k$ and acts inside the larger joint particle-number subspace of dimension $d_J$. We write $\mathfrak{so}(d)=\{X\in\mathbb{R}^{d\times d}:X^\top=-X\}$ for the real-orthogonal Lie algebra; it embeds in the Hermitian operators on $\mathbb{C}^d$ as $i\,\mathfrak{so}(d)$. The basis of choice on the embedding register is the lex-ordered particle-number-$k$ family $\{\ket{S}: \abs{S} = k,\, S \subseteq [D] \}$, with each $\ket{S}$ corresponding to a computational-basis state with the bits indexed by $S$ set to one. The trainable gates are particle-number-preserving Givens rotations on this basis. We write $\langle A,B\rangle_{\mathrm{HS}}=\Tr(A^\dagger B)$ for the Hilbert--Schmidt inner product on operators, $\|A\|_F^2=\Tr(A^\top A)$ for the Frobenius norm of a real matrix, and $\sigma(z)=(1+e^{-z})^{-1}$ for the logistic sigmoid. The Pauli operators on a single qubit are denoted $X, \, Y, \, Z$; expressions such as $Z_{a} \, Z_{b}$ and $X_{i} \, X_{j} + Y_{i} \, Y_{j}$ refer to the corresponding two-qubit Pauli observables on the indexed qubits.

The Lie-algebraic theorem of Ragone et al.~\cite{ragone2024lie} bounds the loss variance of a parametrised cost in terms of the dynamical Lie algebra (DLA) $\mathfrak{g}$ of the circuit and the Hilbert--Schmidt projections of $\rho$, of the parameter generator, and of the observable $O$ onto its simple ideals. The hypothesis is the Lie-Algebra-Supported Ansatz (LASA) of Ragone et al.\ (Definition~2.4): both $i\rho \in \mathfrak{g}$ and $iO \in \mathfrak{g}$. For our DLA $\mathfrak{g}=\mathfrak{so}(d_E)$, the Hermitian image $i\,\mathfrak{so}(d_E)$ consists of pure-imaginary off-diagonal antisymmetric matrices in the particle-number basis. The readouts in our architecture are real-symmetric in this basis: the edge-prediction observable $Z_aZ_b$ and the $1$-RDM diagonal $Z_i$ are diagonal (a special case of real-symmetric), and the $1$-RDM off-diagonal $X_iX_j+Y_iY_j$ is real-symmetric off-diagonal. For any real-symmetric $O$ and any antisymmetric $Y$, the Hilbert--Schmidt overlap satisfies
\begin{equation}\label{eq:hs-zero}
  \langle O, iY\rangle_{\mathrm{HS}} \;=\; i\,\Tr(OY) \;=\; i\,\Tr\bigl((OY)^\top\bigr) \;=\; i\,\Tr(Y^\top O^\top) \;=\; -i\,\Tr(YO) \;=\; -i\,\Tr(OY),
\end{equation}
using $\Tr A=\Tr A^\top$, $(AB)^\top=B^\top A^\top$, $O^\top=O$, $Y^\top=-Y$, and cyclicity. The relation $\Tr(OY)=-\Tr(OY)$ forces $\Tr(OY)=0$, hence $iO\notin\mathfrak{g}$, and the LASA bound returns a vacuous zero.

The particle-number-preserving theorem of Monbroussou et al.~\cite{monbroussou2025trainability} (Theorem~3) gives an exact closed form for the gradient variance of a single-register Givens circuit on $n$ qubits in the particle-number-$k$ subspace of dimension $d_k=\binom{n}{k}$,
\begin{equation}\label{eq:monbroussou}
  \mathbb{E}_{\zeta^0, y}\,\Var_\theta\!\bigl[\partial_\lambda C(\theta)\bigr] \;=\; \frac{8\,k\,(n-k)}{n\,(n-1)\,d_k},
\end{equation}
with $\zeta^0$ the input-state amplitudes, $y$ the target amplitudes, and $C(\theta)=\|\zeta-y\|_2^2$ the squared-Euclidean amplitude cost. The hypothesis (Eq.~22 of \cite{monbroussou2025trainability}) requires inputs and targets each with zero mean and isotropic squared moment $1/d_k$. Our cost is binary cross-entropy on a deterministic structured input.

We sidestep both gaps by integrating directly over the Haar measure on $\mathrm{SO}(d)$ in its defining representation. Let $E_{ab}$ denote the matrix with entry $1$ at row $a$, column $b$ and zero elsewhere, and define a Givens generator of $\mathfrak{so}(d)$,
\begin{equation}\label{eq:gen}
  H_p \;=\; E_{ab} - E_{ba}, \qquad \|H_p\|_F^2 \;=\; 2, \qquad H_p^2 \;=\; -(E_{aa} + E_{bb}).
\end{equation}
Fix a unit input vector $\ket{\psi}\in\mathbb{R}^d$ and a real-symmetric observable $O\in\mathbb{R}^{d\times d}$, and let $U_L,U_R$ be independent Haar-random rotations in $\mathrm{SO}(d)$. Setting
\begin{equation}\label{eq:phim}
  \ket{\phi} \;=\; U_L\ket{\psi}, \qquad M \;=\; U_R^\top\, O\, U_R,
\end{equation}
$\ket{\phi}$ is uniform on the unit sphere $S^{d-1}\subset\mathbb{R}^d$ by left-invariance of Haar measure. The parameter-shift gradient at $\theta=0$ is then
\begin{equation}\label{eq:gp}
  g_p \;=\; \bra{\phi}\,[M, H_p]\,\ket{\phi}.
\end{equation}
The commutator $[M,H_p]$ is real symmetric: with $M^\top=M$ and $H_p^\top=-H_p$,
\begin{equation}\label{eq:commsym}
  [M,H_p]^\top \;=\; H_p^\top M - M H_p^\top \;=\; -H_p M + M H_p \;=\; [M,H_p],
\end{equation}
and $\Tr[M,H_p]=0$ by cyclicity. Hence $\mathbb{E}[g_p]=0$ and
\begin{equation}\label{eq:sod-haar}
  \Var[g_p]
  \;=\;
  \frac{8}{d(d-1)(d+2)}\,\Bigl\|O-\frac{\Tr(O)}{d}I\Bigr\|_F^2.
\end{equation}

The proof has three steps. \emph{Spherical fourth moment.} For $\phi$ uniform on $S^{d-1}$ and any traceless symmetric $A$,
\begin{equation}\label{eq:sphere4}
  \E_\phi\!\bigl[(\phi^\top A\phi)^2\bigr] \;=\; \frac{2\,\Tr(A^2)}{d(d+2)}.
\end{equation}
\emph{Commutator Frobenius norm.} Since $[M,H_p]$ is symmetric, direct expansion and cyclicity give
\begin{equation}\label{eq:commfrob}
  \|[M,H_p]\|_F^2 \;=\; 2\,\Tr(M H_p M H_p) \;-\; 2\,\Tr(M^2 H_p^2).
\end{equation}
\emph{Orthogonal Weingarten averages.} The Schur decomposition $\mathrm{Sym}(d)=\mathrm{Sym}_0(d)\oplus\mathbb{R}\cdot I$ of real symmetric matrices into traceless plus trace parts, applied to $H_p^2$ (symmetric, $\Tr H_p^2=-2$), gives
\begin{align}
  \E_{U_R}\,\Tr(M^2\, H_p^2) &\;=\; -\frac{2\,\Tr(O^2)}{d}, \label{eq:wein1} \\
  \E_{U_R}\,\Tr(M\, H_p\, M\, H_p) &\;=\; \frac{2}{d-1}\Bigl(\Tr(O^2) - \tfrac{1}{d}\Tr(O)^2\Bigr), \label{eq:wein2}
\end{align}
and hence
\begin{equation}\label{eq:wein-comb}
  \E_{U_R}\,\|[M,H_p]\|_F^2 \;=\; \frac{4}{d-1}\,\Bigl\|O - \tfrac{\Tr O}{d}I\Bigr\|_F^2.
\end{equation}
Combining \eqref{eq:sphere4} and \eqref{eq:wein-comb} yields Eq.~\eqref{eq:sod-haar}. The dimension factor $d(d-1)(d+2)=O(d^3)$ comes from the algebra alone, and the Frobenius factor is bounded above by $\|O\|_F^2$ and below by a positive constant for any non-constant $O$. We verified the closed form at $D=6,k=3,d=20$ for $O=Z_aZ_b$: the formula predicts $\Var=0.01837$ and a Monte-Carlo over $2\times10^5$ Haar SO$(20)$ samples gives $0.01829$, agreement to $0.5\%$. We further verified at $d\in\{5,10,20,50\}$ with random real-symmetric $O$, agreement within $0.7\%$ in each case.

The trainable evolution layer $W$ acts on the embedding register through a sequence of Givens rotations. Two connectivities appear in our experiments, and they have different DLAs. In the \emph{nearest-neighbour variant}, Givens act between adjacent qubits in successive layers of $D-1$ gates each. The DLA is $\mathfrak{so}(D)$, of dimension $D(D-1)/2$, acting on the particle-number-$k$ subspace through the $k$-th compound (exterior power) representation $C^{(k)}(\mathrm{SO}(D))$. Eq.~\eqref{eq:sod-haar} at $d=D$ gives a per-parameter floor of order $1/D^3$. In the \emph{Ehrlich (all-pair) variant}, controlled-Givens rotations act between every pair of particle-number-$k$ basis states that differ by one excitation, in the order prescribed by the Ehrlich Gray code over $k$-subsets~\cite{Farias2025}. The DLA is the full $\mathfrak{so}(d_E)$, of dimension $d_E(d_E-1)/2$, on the same subspace, and Eq.~\eqref{eq:sod-haar} at $d=d_E$ gives a per-parameter floor of order $1/d_E^3$. Connectivity choice trades off two competing effects: the Ehrlich variant covers a larger subgroup and is the more expressive subcircuit, but its DLA is large, and the per-parameter floor is correspondingly lower. The nearest-neighbour variant sits inside this larger group as a structured subgroup; its DLA is smaller, and the floor is higher, at the cost of a smaller search space. Monbroussou et al.~\cite{monbroussou2025trainability} (Section~5) identify this as the regime where the dimension-only heuristic of \cite{Larocca2023} fails: dimension alone does not determine variance scaling for particle-number-preserving architectures, and orbit structure on the particle-number subspace also matters.

The numerical values predicted by Eq.~\eqref{eq:sod-haar} for a Haar-random rotation of the whole particle-number-$k$ subspace ($d=d_E$, the group of the Ehrlich variant) at the readout $O=Z_aZ_b$, restricted to half-filling, exhibit the polynomial scaling and a margin over the unstructured-circuit baseline that grows with $D$:
\begin{center}\small
\begin{tabular}{ccrrrr}
\toprule
$(D,k)$ & $d_E$ & $\|O-\tfrac{\Tr(O)}{d_E}I\|_F^2$ & $\mathrm{Var}\,[\partial_p\mathcal{L}]$ & $4^{-D}$ & ratio \\
\midrule
$(6,3)$ & $20$ & $19.20$ & $1.84\times 10^{-2}$ & $2.44\times 10^{-4}$ & $75$ \\
$(8,4)$ & $70$ & $68.57$ & $1.58\times 10^{-3}$ & $1.53\times 10^{-5}$ & $103$ \\
$(10,5)$ & $252$ & $248.89$ & $1.24\times 10^{-4}$ & $9.54\times 10^{-7}$ & $130$ \\
\bottomrule
\end{tabular}
\end{center}
The same scaling holds for the $1$-RDM diagonal $Z_i$ and the off-diagonal $X_iX_j+Y_iY_j$, with traceless Frobenius norms differing only by an $O(1)$ factor at fixed $(D,k)$. Polynomial-cost Hartree--Fock and matchgate-shadow readouts inherit this scaling because their estimators are unbiased linear combinations of single-particle observables, and a fixed Lipschitz classical head composed onto the resulting feature vector preserves polynomial gradient variance up to a constant prefactor.

The mixing layer $M$ acts on the joint particle-number subspace of total weight $j+k$, of dimension $d_J=\binom{N+D}{j+k}$, polynomial in $N$ at fixed $(D,k,j)$. Under all-pair connectivity within the joint subspace, the generator $H_{ab}=E_{ab}-E_{ba}$ for each pair $1\le a<b\le N+D$ acts as a real antisymmetric operator on the basis of $(j+k)$-subsets of $[N+D]$, mediating single-mode-swap transpositions $S\leftrightarrow(S\setminus\{a\})\cup\{b\}$ when exactly one of $a,b$ lies in $S$; its Lie closure is contained in $\mathfrak{so}(d_J)$. The connectivity graph induced on these subsets by single transpositions is connected. On the single-particle subspace ($j+k=1$, dimension $N+D$), the commutator of two adjacent generators reduces exactly to
\begin{equation}\label{eq:hcomm}
  [H_{ab}, H_{bc}] \;=\; -H_{ac}.
\end{equation}
On the higher-weight subspaces ($j+k\ge 2$), which is the relevant regime for our architecture, the closure depends on the gate. For two-qubit Givens rotations as gates (\texttt{SingleExcitation}, which carries no Jordan--Wigner string), $H_{ab}$ moves one particle without a sign and the closure is the full $\mathfrak{so}(d_J)$. For fermionic hopping generators, which carry the Jordan--Wigner sign of the occupied modes between $a$ and $b$, the closure is only the compound representation of $\mathfrak{so}(N+D)$, of dimension $(N+D)(N+D-1)/2$. The standard density argument for connected single-particle transpositions on a finite-dimensional real-orthogonal irrep then gives equality to the full $\mathfrak{so}(d_J)$~\cite{dalessandro2007introduction}. We verified this directly by closing the algebra under commutation: starting from the generator set, computing all pairwise commutators, retaining the ones not already in the span, and repeating until the dimension stopped growing. The four cases tested are $(N,D,j,k)\in\{(3,3,1,1),(4,3,1,1),(2,4,1,2),(3,3,1,2)\}$, with computed DLA dimensions $\{105,210,190,190\}$ for the unsigned generators, all matching $d_J(d_J-1)/2$ exactly; with Jordan--Wigner signs, and with nearest-neighbour connectivity, the closures have dimension $15$, $21$, $15$ and $15$, that is $(N+D)(N+D-1)/2$. The same proof construction applies at the configuration $(5,4,1,2)$ used elsewhere in the paper, where $d_J=84$ and $\dim\mathfrak{so}(84)=3486$ (not computed by closure). The joint mixing layer of the simulated TSP model is a nearest-neighbour pyramid of Givens rotations, whose closure is the smaller $\mathfrak{so}(N+D)$; the values at $d=d_J$ below describe an all-pair mixer.

Eq.~\eqref{eq:sod-haar} at $d=d_J$ gives a per-parameter floor of order $1/d_J^3$. At fixed $(D,k,j)$ this is $\Omega(1/N^{3(j+k)})$, polynomial in $N$. For the Euclidean TSP regime $j=1,(D,k)=(6,3)$, the polynomial degree is $12$. Tabulated alongside the unstructured baseline $4^{-(N+D)}$, the joint-register bound shows the polynomial-versus-exponential separation that is the trainability story of this architecture:
\begin{center}\small
\begin{tabular}{rrrrr}
\toprule
$N$ & $d_J$ & $4^{-(N+D)}$ & $8/[d_J(d_J-1)(d_J+2)]$ & ratio \\
\midrule
$5$ & $330$ & $2.4\times 10^{-7}$ & $2.2\times 10^{-7}$ & $0.93$ \\
$10$ & $1820$ & $2.3\times 10^{-10}$ & $1.3\times 10^{-9}$ & $5.7$ \\
$20$ & $14950$ & $2.2\times 10^{-16}$ & $2.4\times 10^{-12}$ & $1.1\times 10^{4}$ \\
$30$ & $58905$ & $2.1\times 10^{-22}$ & $3.9\times 10^{-14}$ & $1.9\times 10^{8}$ \\
\bottomrule
\end{tabular}
\end{center}
The polynomial floor sits below the exponential baseline at $N=5$ because the latter is not yet small at $N+D=11$; the separation grows quickly thereafter, reaching four orders of magnitude at $N=20$ and eight orders at $N=30$. The same analysis at $(D,k)=(8,4),j=1,N=20$ gives $d_J=98{,}280$, polynomial floor $8.4\times 10^{-15}$, baseline $1.4\times 10^{-17}$, separation $610$, confirming that the polynomial-versus-exponential margin grows with both $N$ and $D$.

Eq.~\eqref{eq:sod-haar} bounds the gradient of an expectation $\langle O\rangle$. The BCE cost on edge probabilities is
\begin{equation}\label{eq:bce}
  \mathcal{L}_{\mathrm{BCE}}(z) \;=\; -y\log\sigma(z) \;-\; (1-y)\log\bigl(1 - \sigma(z)\bigr), \qquad z = \langle Z_a Z_b\rangle \in [-1, 1] \, ,
\end{equation}
and direct differentiation gives the chain rule
\begin{equation}\label{eq:bce-grad}
  \partial_z \mathcal{L}_{\mathrm{BCE}} \;=\; \sigma(z) - y, \qquad \partial_p \mathcal{L}_{\mathrm{BCE}} \;=\; \bigl(\sigma(z) - y\bigr)\,\partial_p z.
\end{equation}
The factor $\sigma(z)-y$ lies in $[-1,1]$, and its expected square at random initialisation is bounded below by a positive constant for $y\in\{0,1\}$ and $\sigma(z)$ not concentrated at $y$. The expected variance of the BCE gradient is therefore bounded above by $\Var[\partial_p z]$ and, in expectation over labels and trajectories, below by a positive constant times the same. The argument extends without modification to any fixed Lipschitz classical head composed onto the $1$-RDM features.

\begin{siproposition}[Polynomial trainability of the two-register QGNN]\label{prop:trainability_si}
Consider the two-register QGNN of \mainref{ssec:results_quantum}, with embedding-register evolution $W$ in either the nearest-neighbour or the Ehrlich (all-pair) connectivity, and joint-register mixing $M$ with all-pair connectivity. Suppose the cost is binary cross-entropy on edge probabilities or any fixed Lipschitz classical head on the $1$-RDM features, and the trainable circuit forms an approximate $2$-design on each of $\mathrm{SO}(D)$ (or $\mathrm{SO}(d_E)$ in the Ehrlich variant) and $\mathrm{SO}(d_J)$ for the joint-register parameters. Then, for any single gate parameter $\theta_p$, the gradient variance at random initialisation, in expectation over labels, satisfies
\[
  \mathrm{Var}_\theta\bigl[\partial_{\theta_p}\mathcal{L}\bigr]
  \;=\;
  \Omega\!\left(\frac{1}{\mathrm{poly}(N,D)}\right),
\]
polynomial of degree at most $\max\{4,\,3(j+k)\}$ at fixed $(j,k)$.
\end{siproposition}

\begin{proof}
A $W$-parameter is bounded by the per-register form of Eq.~\eqref{eq:sod-haar} at $d=D$ (nearest-neighbour) or $d=d_E$ (Ehrlich), each polynomial in $D$. An $M$-parameter is bounded by the joint-register form at $d=d_J=\Theta(N^{j+k})$, polynomial in $N$ at fixed $(D,k,j)$. The minimum of two polynomial floors is polynomial. The BCE chain-rule bound and any fixed Lipschitz classical head preserve polynomial scaling in expectation.
\end{proof}

Three remarks on the assumptions of \autoref{prop:trainability_si} are worth making. The proposition is stated at random initialisation, which is the worst case for arguments of this kind: the Haar-averaged variance integrates over the entire parameter space, including configurations far from the trained one. In practice, parameters are initialised near zero with a small standard deviation, so each Givens rotation begins close to the identity, and the trainable circuit applies a small perturbation to the initial state. Near-identity initialisation gives strictly larger gradients than the Haar-averaged worst case for shallow effective depth, by a separate mechanism that does not invoke a $2$-design hypothesis at all, and our experiments use this regime throughout. The $2$-design hypothesis itself is provable for circuit depths of order $d \log d$ in the defining representation~\cite{Larocca2023}, and our typical depths sit in the same regime up to constants; the empirical near-neighbour-versus-Ehrlich variance ratio at $(D,k)=(6,3)$ is consistent with mild under-saturation at moderate $D$. Finally, the all-pair connectivity assumption for $M$ is the assumption used in our closure verification, but the implementation uses a nearest-neighbour pyramid on the joint register; for that connectivity the DLA may be a proper subalgebra of $\mathfrak{so}(d_J)$, in which case \autoref{prop:trainability_si} applies at the effective dimension of the subalgebra rather than at $d_J$. The polynomial-in-$N$ scaling at fixed $(D,k,j)$ is preserved as long as this effective dimension grows polynomially in $N$, which we flag as the most pressing direct verification item.

A complementary observation, independent of these conditional points: in our experiments, single-step transfer initialisation from the $N{=}5$ model to a larger graph preserves $10$ to $240\times$ larger gradients than random initialisation across $N\in[5,20]$ at fixed $(D,k,j)$; the chained arm, re-initialised from $N{-}1$ at every step, does not (\autoref{fig:si_shots_extrap}). This is empirical evidence that the polynomial floor is operationally meaningful in practice, not just an asymptotic statement.

\medskip\noindent\textbf{Empirical verification.} The same data underlying the main-text \mainref*{fig:trainability} is reported at two additional resolutions in the figures below. \autoref{fig:si_trainability_3panel} shows two diagnostic sweeps: $N$-dependence at fixed $(D,k){=}(6,3)$ across $N\in\{5,8,10,12,15,20,25,30,40,50\}$ (panel a) and $D$-scaling at $N{=}10$, $k{=}D/2$ for $D\in\{4,6,8,10,12,14\}$ against the $1/\dim(\mathfrak{g})^2$ Ragone--Larocca prediction and the $1/\binom{D}{k}^2$ Monbroussou bound (panel b). The mean per-parameter gradient variance over quantum-side parameters declines slowly in $N$, from $7.5\times10^{-4}$ at $N{=}5$ to $4.3\times10^{-4}$ at $N{=}50$ (log-log slope $-0.21$). In $D$ it falls as $\binom{D}{k}^{-1.08}$ across three orders of magnitude in $\binom{D}{k}$. It lies below the Ragone--Larocca value at every $D$, by a factor of $7$ at $D{=}6$ and $45$ at $D{=}14$. It lies below the Monbroussou value for $D\le 8$ and above it from $D{=}10$, by a factor of $32$ at $D{=}14$.

\begin{figure}[h]
\centering
\includegraphics[width=0.95\linewidth]{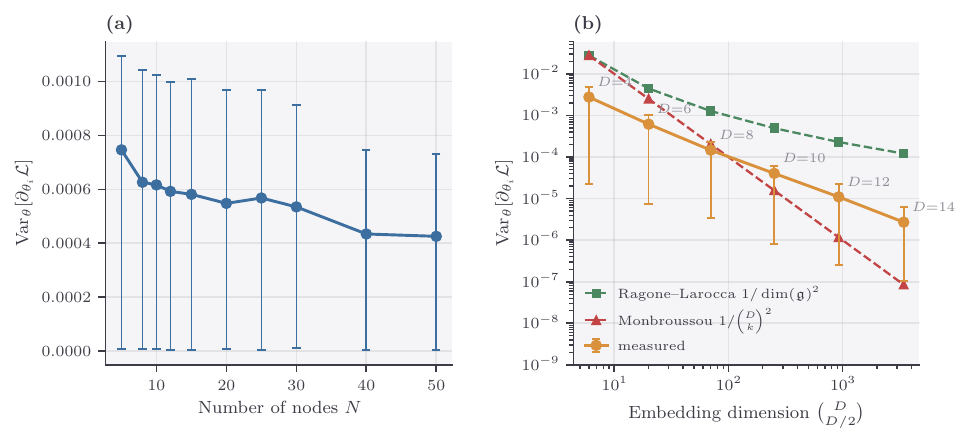}
\caption{\textbf{Gradient variance against $N$ and $D$.} (a)~$N$-dependence at fixed $(D,k){=}(6,3)$, $j{=}1$: mean per-parameter gradient variance over quantum-side parameters at random initialisation, $200$ initialisations per point; bars span the minimum and maximum over the quantum parameter groups. The variance declines slowly, from $7.5\times10^{-4}$ at $N{=}5$ to $4.3\times10^{-4}$ at $N{=}50$. (b)~$D$-scaling at $N{=}10$, $k{=}D/2$, $j{=}1$, against the Ragone--Larocca $1/\dim(\mathfrak{g})^2$ bound and the Monbroussou $1/\binom{D}{k}^2$ bound. The measured variance falls as $\binom{D}{k}^{-1.08}$. It is below the Ragone--Larocca value at every $D$, below the Monbroussou value for $D\le8$, and above it for $D\ge10$.}
\label{fig:si_trainability_3panel}
\end{figure}

\begin{table}[h]
\centering\small
\caption{\textbf{Numerical values for the trainability sweeps of \autoref{fig:si_trainability_3panel}.} $200$ random initialisations per point. Left: $N$-dependence at fixed $(D,k){=}(6,3)$, $j{=}1$. Right: $D$-scaling at $N{=}10$, $k{=}D/2$, $j{=}1$, with the Ragone--Larocca lower bound $1/\dim(\mathfrak{g})^2$ at $\dim(\mathfrak{g}){=}D(D{-}1)/2$ and the Monbroussou bound $1/\binom{D}{k}^2$ at $\binom{D}{k}$. Variances are mean per-parameter $\mathrm{Var}_\theta[\partial_{\theta_i}\mathcal{L}]$ over quantum-side parameters.}
\label{tab:si_trainability}
\begin{minipage}[t]{0.42\linewidth}
\centering
\begin{tabular}{@{}r r@{}}
\toprule
$N$ & $\mathrm{Var}_\theta[\partial_{\theta}\mathcal{L}]$ \\
\midrule
$5$  & $7.47\times 10^{-4}$ \\
$8$  & $6.26\times 10^{-4}$ \\
$10$ & $6.17\times 10^{-4}$ \\
$12$ & $5.93\times 10^{-4}$ \\
$15$ & $5.82\times 10^{-4}$ \\
$20$ & $5.48\times 10^{-4}$ \\
$25$ & $5.68\times 10^{-4}$ \\
$30$ & $5.35\times 10^{-4}$ \\
$40$ & $4.35\times 10^{-4}$ \\
$50$ & $4.26\times 10^{-4}$ \\
\bottomrule
\end{tabular}
\end{minipage}\hfill
\begin{minipage}[t]{0.55\linewidth}
\centering
\begin{tabular}{@{}r r r r r@{}}
\toprule
$D$ & $\binom{D}{k}$ & $\mathrm{Var}_\theta[\partial_{\theta}\mathcal{L}]$ & $1/\dim(\mathfrak{g})^2$ & $1/\binom{D}{k}^2$ \\
\midrule
$4$  & $6$    & $2.76\times 10^{-3}$ & $2.78\times 10^{-2}$ & $2.78\times 10^{-2}$ \\
$6$  & $20$   & $6.17\times 10^{-4}$ & $4.44\times 10^{-3}$ & $2.50\times 10^{-3}$ \\
$8$  & $70$   & $1.46\times 10^{-4}$ & $1.28\times 10^{-3}$ & $2.04\times 10^{-4}$ \\
$10$ & $252$  & $4.04\times 10^{-5}$ & $4.94\times 10^{-4}$ & $1.57\times 10^{-5}$ \\
$12$ & $924$  & $1.10\times 10^{-5}$ & $2.30\times 10^{-4}$ & $1.17\times 10^{-6}$ \\
$14$ & $3{,}432$ & $2.71\times 10^{-6}$ & $1.21\times 10^{-4}$ & $8.49\times 10^{-8}$ \\
\bottomrule
\end{tabular}
\end{minipage}
\end{table}

The per-parameter gradient variance also resolves at the gate-group level. \autoref{fig:si_var_groups} reports the mean and the maximum per-parameter gradient variance, separated into seven parameter groups (the three final-pyramid blocks, the compound layers, the embedding MLPs, the adjacency MLP and the readout head) under the four initialisation regimes of panel (b) of \mainref{fig:trainability}. Under random initialisation the gradient variance decays polynomially in $N{+}D$ for six of the seven groups, with exponents from $-0.9$ (compound layers) to $-2.8$ (city--city pyramid); the adjacency-angle MLP is the exception, its variance rising mildly from $4\times10^{-4}$ to $1\times10^{-3}$. No group decays exponentially. Transfer from the $N{=}5$ model lifts every group above its random-initialisation value, and chained transfer collapses every group at large $N$.

\begin{figure}[h]
\centering
\includegraphics[width=0.95\linewidth]{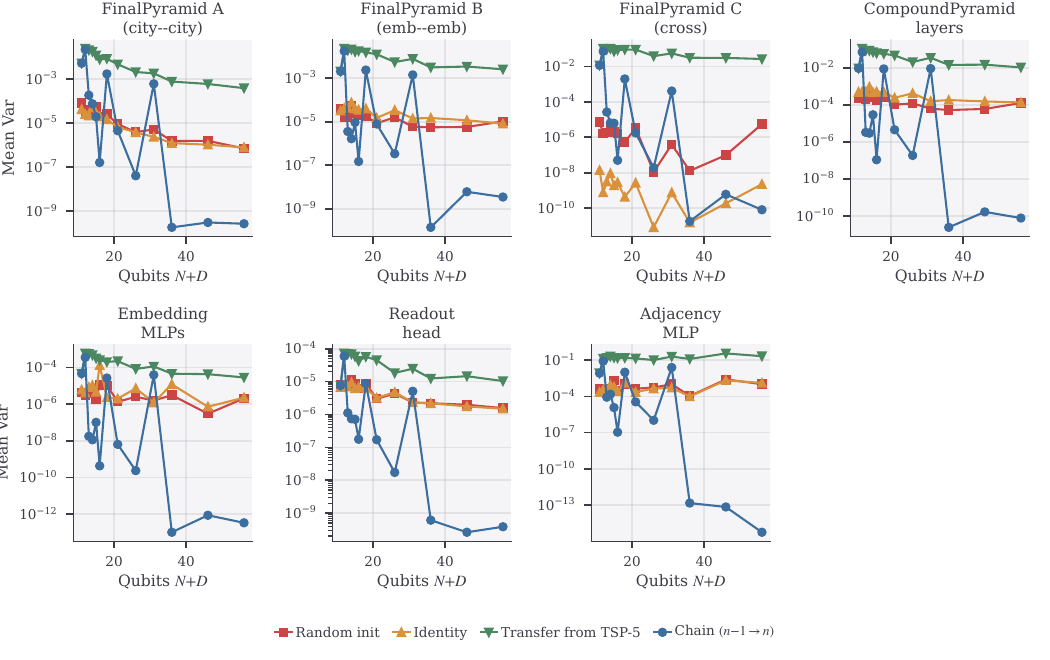}\\[0.4em]
\includegraphics[width=0.95\linewidth]{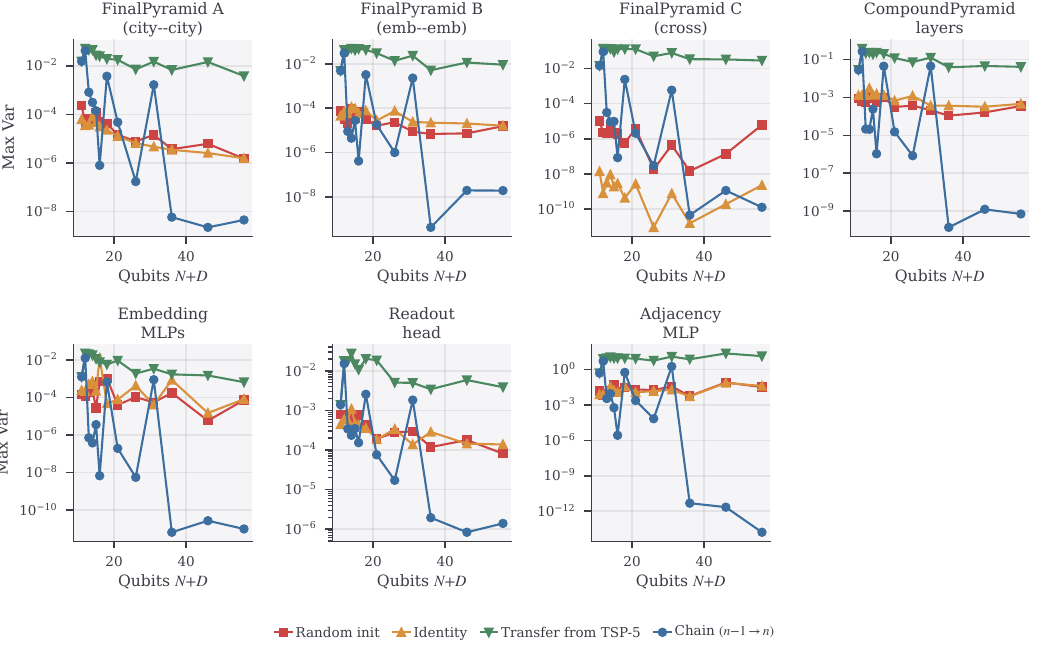}
\caption{\textbf{Gradient variance ordering across trainable gate groups.} Per-parameter gradient variance, separated by gate group. Top: mean over the parameters of each group. Bottom: maximum over the parameters of each group. Each panel shows one parameter group under the four initialisation regimes of \mainref*{fig:trainability}b. Under random initialisation, exponents in $N{+}D$ range from $-2.8$ to $+0.5$; no group decays exponentially.}
\label{fig:si_var_groups}
\end{figure}

The decomposition into quantum and classical parameters is reported in \autoref{fig:si_var_qc}. Under random initialisation the quantum-parameter gradient variance decays polynomially, as $(N{+}D)^{-2.5}$, from $1.4\times10^{-4}$ at $N{=}5$ to $5.5\times10^{-6}$ at $N{=}50$, while the classical-parameter variance is flat at about $1.5\times10^{-5}$; the two cross between $N{=}30$ and $N{=}40$. Transfer from the trained $N{=}5$ model and identity initialisation show the same polynomial decay on the quantum side ($(N{+}D)^{-2.4}$ and $(N{+}D)^{-2.6}$), and transfer keeps the quantum-parameter variance one to three orders of magnitude above the random start at every size. Chained transfer collapses on both sides (see panel (b) of \mainref{fig:trainability}).

\begin{figure}[h]
\centering
\includegraphics[width=0.95\linewidth]{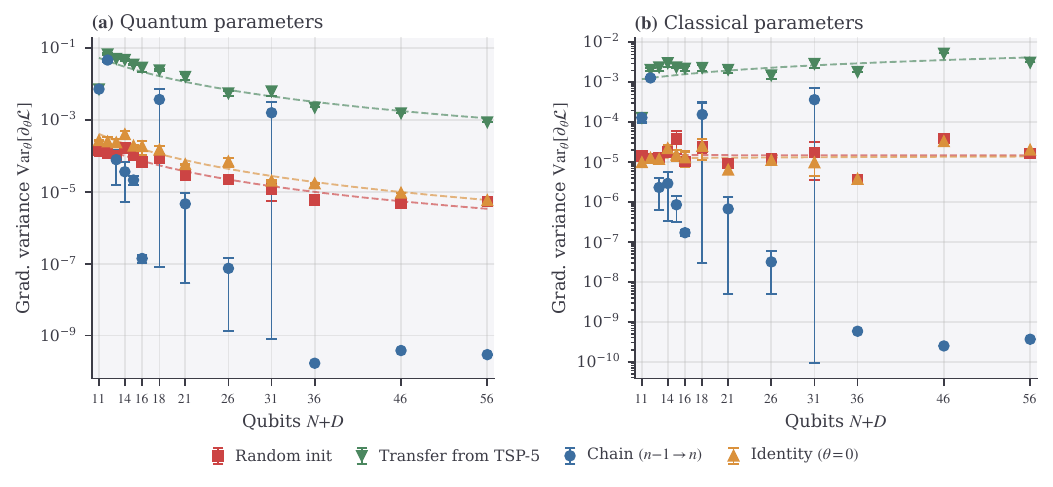}
\caption{\textbf{Quantum- and classical-parameter gradient variance.} Per-parameter gradient variance against qubit count $N{+}D$ at $(D,k){=}(6,3)$ for (a)~the quantum (Givens-rotation) parameters and (b)~the classical (embedding, adjacency and readout) parameters, under the four initialisation regimes of \mainref*{fig:trainability}b. Points are seed means with min--max bars; dashed lines are power-law fits (none for chained transfer). The quantum-parameter variance decays polynomially under random, identity and transfer initialisation; the classical-parameter variance is flat or rising.}
\label{fig:si_var_qc}
\end{figure}

\autoref{fig:si_shots_extrap} reports the shots required to estimate the gradient to relative accuracy $\delta{=}0.05$ as a function of $N+D$, at $(D,k){=}(6,3)$ with parameter-shift differentiation, for the four initialisation regimes of panel (b) of \mainref{fig:trainability}. From random initialisation the cost per gradient step grows polynomially, as $(N{+}D)^{4.1}$, from $1.4\times10^{9}$ shots at $N{+}D{=}11$ to $4.7\times10^{11}$ at $N{+}D{=}56$. Transfer from the trained $N{=}5$ model follows the same power law at a $160$--$680\times$ lower cost, reaching $2.9\times10^{9}$ shots per step at $N{+}D{=}56$. Identity initialisation is within a factor of three of the random start. Chained transfer is cheaper only for $N\le6$; from $N{=}7$ its gradient variance collapses and its cost exceeds every other regime.

\begin{figure}[h]
\centering
\includegraphics[width=0.95\linewidth]{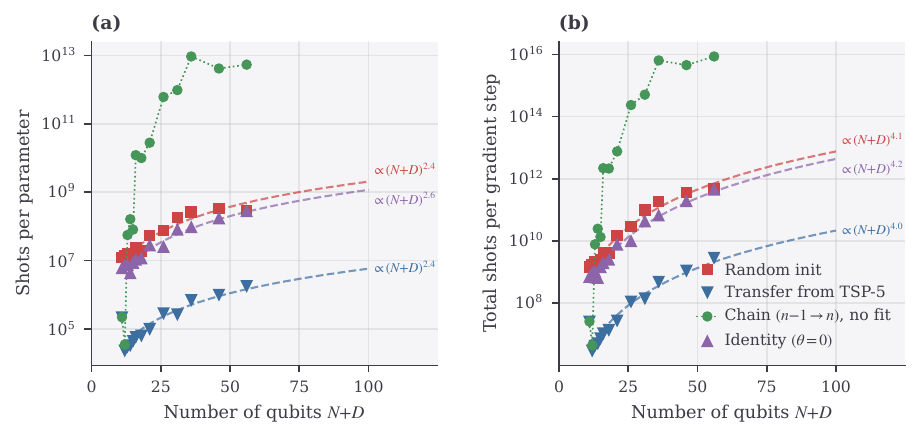}
\caption{\textbf{Shot budget per parameter-shift gradient step.} Shot complexity (a)~per parameter and (b)~per full parameter-shift gradient step at relative accuracy $\delta{=}0.05$, plotted against the qubit count $N{+}D$ at $(D,k){=}(6,3)$. Random initialisation (squares), transfer from TSP-5 (down-triangles) and identity initialisation (up-triangles) grow polynomially, as $(N{+}D)^{2.4}$--$(N{+}D)^{2.6}$ per parameter and $(N{+}D)^{4.0}$--$(N{+}D)^{4.2}$ per step (dashed: power-law fits, extrapolated to $N{+}D{=}100$). Transfer from TSP-5 is $160$--$680\times$ cheaper than random initialisation at every size. Chained transfer (circles, no fit) is cheaper only up to $N{=}6$. Points are seed means over two seeds for $N\le25$ and one seed for $N\ge30$.}
\label{fig:si_shots_extrap}
\end{figure}

\section{Measurement strategy comparison}\label{si:measurement}

Each forward pass of the trained QGNN ends with feature extraction from the joint state on $N+D$ qubits. We compare four candidate strategies on arbitrary input states. \autoref{tab:supp_cost} reports the asymptotic shot complexity required to make every entry of each strategy's native feature vector individually trustable to additive error $\varepsilon$. For HF and the matchgate shadow, where the native feature is the $D \times D$ $1$-RDM matrix and the matchgate-frame channel inverse couples entries, this per-entry budget coincides with the Frobenius bound up to constants; for amplitude, the bin-population convention is what controls deployment.

\begin{table}[h]
\centering\small
\caption{\textbf{Per-strategy shot complexity for full-feature recovery on the embedding register.} Each row gives shots to make every entry of the strategy's native feature vector individually trustworthy to additive error $\varepsilon$. Native feature vectors: amplitude reads $\binom{D}{k}$ probabilities $|c_S|^2$, where the bin-population budget divides across bins; $Z{+}ZZ$ reads $O(D^2)$ comp-basis correlators $\langle Z_p\rangle$, $\langle Z_p Z_q\rangle$ jointly diagonal in $Z$; HF and shadow reconstruct the full $1$-RDM $\gamma$ ($D^2$ entries) under the matchgate channel inverse, which adds a factor of $D$. Per measurement setting on the embedding register, amplitude and $Z{+}ZZ$ require zero additional depth, deterministic HF requires $O(D)$ depth, and the matchgate shadow requires $O(D^2)$ depth.}
\label{tab:supp_cost}
\begin{tabular}{@{}l c@{}}
\toprule
Strategy & Shots (per-entry $\varepsilon$ on native feature) \\
\midrule
Amplitude          & $O\!\left(\binom{D}{k}/\varepsilon^2\right)$ \\
Embedding $Z{+}ZZ$ & $O(D^2/\varepsilon^2)$ \\
Deterministic HF   & $O(D^3/\varepsilon^2)$ \\
Matchgate shadow   & $O(D^3/\varepsilon^2)$ \\
\bottomrule
\end{tabular}
\end{table}

We work on the embedding-register Hilbert space of fixed particle number $k$ on $D$ qubits. For a real-amplitude state $\ket{\psi} = \sum_S c_S \ket{S}$, the $1$-RDM is $\gamma_{pq} = \bra{\psi} a_p^\dagger a_q \ket{\psi}$ in the Jordan--Wigner mapping; $\gamma$ is real symmetric with $\Tr\gamma = k$. All four strategies share an outer projection step on the node register: each readout first measures the node register in the comp basis, accepts only outcomes whose lit set has cardinality exactly $j$, identifies that lit set with a $j$-subset $T \subseteq [N]$, and runs the embedding-register protocol on the post-measurement conditional state $\ket{\phi_T}_{\rm emb}$ on the particle-number-$k$ subspace.

The four native feature vectors carry different information about the embedding-register state. Amplitude and $Z{+}ZZ$ are confined to comp-basis observables: amplitude resolves the multinomial counts $\hat p_S$, from which $\widehat\gamma_{pp} = \sum_{S \ni p} \hat p_S$ but no off-diagonal entry can be recovered; $Z{+}ZZ$ adds the diagonal $2$-RDM correlators $\langle \hat n_p \hat n_q\rangle$ but is similarly blind to $\gamma_{pq}$ for $p\ne q$. The off-diagonal entries are sign-coherent sums of amplitude products $c_S c_{S^{q\to p}}$ over basis pairs differing by a single mode swap, and the Born statistics $|c_S|^2$ retain no sign information.

A worked example is illustrative. At $(D,k) = (4,1)$ the two particle-number-$1$ states
\[
\ket{\psi_+} = \tfrac12 (\ket{1000} + \ket{0100} + \ket{0010} + \ket{0001}), \qquad
\ket{\psi_-} = \tfrac12 (\ket{1000} - \ket{0100} + \ket{0010} - \ket{0001})
\]
satisfy $\langle Z_p\rangle = \tfrac12$ and $\langle Z_pZ_q\rangle = 0$ for all $p,q$ in both, so the comp-basis-only readouts produce identical native features on the two states. Their $1$-RDMs are orthogonal rank-$1$ projectors with Frobenius distance $\sqrt{2}$, so a downstream task that depends on $\gamma$ cannot distinguish $\ket{\psi_+}$ from $\ket{\psi_-}$ through comp-basis-only features at any shot count. The matchgate-tailored protocols of the next two rows distinguish the two states at finite shot count and are the focus of the rest of this note.

Deterministic Hartree--Fock~\cite{arute2020hartree} and the matchgate shadow~\cite{wan2023matchgate} apply a matchgate basis-change unitary on the embedding register before comp-basis measurement, mapping off-diagonal $\gamma_{pq}$ entries into accessible occupation differences. Both protocols are unbiased on every Hermitian symmetric one-body observable.

The two protocols differ in the choice of basis change. Deterministic HF uses $D/2 + 1$ fixed Givens settings, each at depth $O(D)$, covering all $\binom{D}{2}$ off-diagonal entries. The matchgate shadow draws a fresh Haar-random matchgate $C^{(k)}(U)$ with $U \in \mathrm{SO}(D)$ per snapshot at depth $O(D^2)$, then inverts the resulting channel via Eq.~\eqref{eq:shadow_estimator}.

The two protocols also differ in their per-shot variance. Deterministic HF splits the shot budget across $D/2 + 1$ settings, so the per-entry variance of $\gamma$ from one shot is $(D/2 + 1)/N_{\rm shots}$. The matchgate-shadow channel inverse multiplies each snapshot entry by $(D+2)/2$, so the per-entry variance is $(2D-1)/N_{\rm shots}$. Summing over the $D(D+1)/2$ unique entries of the symmetric matrix gives Frobenius shots $D(D+1)(D+2)/(4\varepsilon^2)$ for HF and $D(D+1)(2D-1)/(2\varepsilon^2)$ for the shadow, both $O(D^3/\varepsilon^2)$ with a constant prefactor gap of about $4$ in HF's favour.

The wall-clock budget on hardware adds the projection multiplier $1/P_{\rm proj}$ from the outer node-register conditioning. The per-iteration operations $V$, $\mathcal{A}$, and $W$ each preserve the node-register particle number $j$ individually, so before the final mixing layer $M$, the projection succeeds on every measurement run. The mixing layer $M$ preserves only the total weight $j+k$, so after $M$ the joint state has support on every $(n_{\rm node}, n_{\rm emb})$ partition with sum $j + k$, and the projection succeeds only on the $(j, k)$ component. For a Haar-random joint state in the particle-number-$(j+k)$ subspace,
\begin{equation}
P_{\rm proj} = \frac{\binom{N}{j} \binom{D}{k}}{\binom{N+D}{j+k}}, \qquad \frac{1}{P_{\rm proj}} \sim \frac{j!}{\binom{D}{k}\,(j+k)!}\,N^k \quad (N \to \infty).
\label{eq:proj_success}
\end{equation}
The total wall-clock cost over all $N$ nodes is therefore $O(N^{k+1} D^3/\varepsilon^2)$ for both polynomial-cost protocols.

The matchgate-shadow estimator is derived as follows. Apply a Haar-random matchgate $C^{(k)}(U)$ before each comp-basis measurement and accumulate $U^\top \mathrm{diag}(\hat n_S) U$ where $\hat n_{S,p} = \mathbf{1}[p \in S]$. The expected snapshot, by joint expectation over Haar $\mathrm{SO}(D)$ and the comp-basis outcome $S$ sampled from $|\bra{S} C^{(k)}(U)\ket{\psi}|^2$, is
\begin{equation}
\mathcal{M}(\gamma)_{pq} = \E_{U,S}\bigl[(U^\top \mathrm{diag}(\hat n_S) U)_{pq}\bigr] = \frac{2}{D+2}\,\gamma_{pq} + \frac{k}{D+2}\,\delta_{pq},
\label{eq:shadow_channel}
\end{equation}
where the inner expectation $\E_S[\hat n_{S,p} \mid U] = (U \gamma U^\top)_{pp}$ uses the homomorphism property of the compound representation acting on one-body observables, and the Haar average on $\mathrm{SO}(D)$ decomposes $\mathrm{Sym}(\mathbb{R}^D) = \mathrm{Sym}_0(\mathbb{R}^D) \oplus \mathbb{R}\cdot I$ and acts as the scalar $2/(D+2)$ on the traceless symmetric component and as the identity on $\mathbb{R}\cdot I$. Inverting on each component gives the unbiased estimator
\begin{equation}
\widehat\gamma = \frac{D+2}{2}\,U^\top\mathrm{diag}(\hat n_S)\,U - \frac{k}{2}\,I,
\label{eq:shadow_estimator}
\end{equation}
with $\E[\widehat\gamma] = \gamma$ and $\Tr(\widehat\gamma) = k$ as an algebraic identity per snapshot. The per-entry variance of a single shadow snapshot is bounded by $2D-1$~\cite{wan2023matchgate}.

A classical readout head $g$ that is Lipschitz-$L$ in the per-node $1$-RDM features propagates a per-entry estimation error $\varepsilon$ on $\widehat\gamma$ to a Frobenius error $\|\widehat\gamma - \gamma\|_F \le D\varepsilon$, and through $g$ to an edge-logit perturbation $|\Delta z_{ij}| \le L D \varepsilon$. Correct decoding of the predicted output requires $|\Delta z_{ij}| < \Delta_{\min}/2$, where $\Delta_{\min}$ is the smallest gap between selected and rejected logits; combined with a union bound over the $D^2$ entries and the $\binom{N}{2}$ edges, the per-readout sample complexity is
\begin{equation}
N_{\rm shots} \ge \frac{8 L^2 D^3 \log(DN/\delta)}{\Delta_{\min}^2}, \qquad
N_{\rm shots}^{\rm wall} \ge \frac{N_{\rm shots}}{P_{\rm proj}},
\label{eq:tour_shots}
\end{equation}
for correct decoding with confidence $1-\delta$.

\begin{siproposition}[Polynomial-cost $1$-RDM readout]
\label{prop:rdm_readout_full}
For the two-register QGNN with matchgate trainable evolution, the embedding-register $1$-RDM conditioned on a $j$-subset $T$ of node-register qubits can be estimated to Frobenius additive error $\varepsilon$ in $O(D^3/\varepsilon^2)$ shots per accepted projection round at circuit depth $O(D)$ per measurement setting via the deterministic protocol of~\cite{arute2020hartree}, or in $O(D^3/\varepsilon^2)$ shots per accepted round at circuit depth $O(D^2)$ per snapshot via the matchgate-shadow estimator of Eq.~\eqref{eq:shadow_estimator}. The shadow's per-entry variance is larger than HF's by a factor of about $4$ from the channel-inverse amplification; both costs are independent of $N$ and $k$. The wall-clock measurement budget is the per-round budget multiplied by $1/P_{\rm proj}$ from Eq.~\eqref{eq:proj_success}, polynomial in $N$ at fixed $(D,j,k)$.
\end{siproposition}

\begin{proof}
Both protocols act on the embedding register at fixed particle number $k$, conditioned on the projection onto a $j$-subset of node-register qubits. The deterministic protocol of~\cite{arute2020hartree} estimates each off-diagonal entry as a single bounded $\pm 1$ correlator after an $O(D)$-depth Givens layer; per-setting Hoeffding gives $1/N$ variance per entry, and summing over the $D(D-1)/2$ off-diagonal entries and $D/2+1$ settings yields $O(D^3/\varepsilon^2)$ for Frobenius error $\varepsilon$. The matchgate-shadow estimator of Eq.~\eqref{eq:shadow_estimator}, at circuit depth $O(D^2)$ per snapshot, is unbiased with per-snapshot per-entry variance at most $2D-1$~\cite{wan2023matchgate}; summing over the $D^2$ entries gives the same $O(D^3/\varepsilon^2)$ bound. The shadow's per-entry variance exceeds the deterministic protocol's by the ratio $(2D-1)/(D/2+1)$, about $2.75$ at $D=6$, and a separate channel-inverse prefactor $(D+2)/2=4$ amplifies the shadow's variance at small shot budgets. The per-round cost is independent of $N$ and $k$; the wall-clock budget multiplies it by $1/P_{\rm proj}$ from Eq.~\eqref{eq:proj_success}, polynomial in $N$ at fixed $(D,j,k)$.
\end{proof}

The $1$-RDM compresses $\binom{D}{k}$ amplitudes into a $D \times D$ symmetric matrix with trace $k$; the generic image dimension is $\min(\binom{D}{k} - 1, D(D+1)/2 - 1)$, with one direction lost at half-filling. The detailed Jacobian-rank derivation, including the Hodge identity $\gamma(\star\psi) = I - \gamma(\psi)$, is in \suppref{si:rdm_jacobian}.

\begin{figure}[h]
\centering
\includegraphics[width=0.97\linewidth]{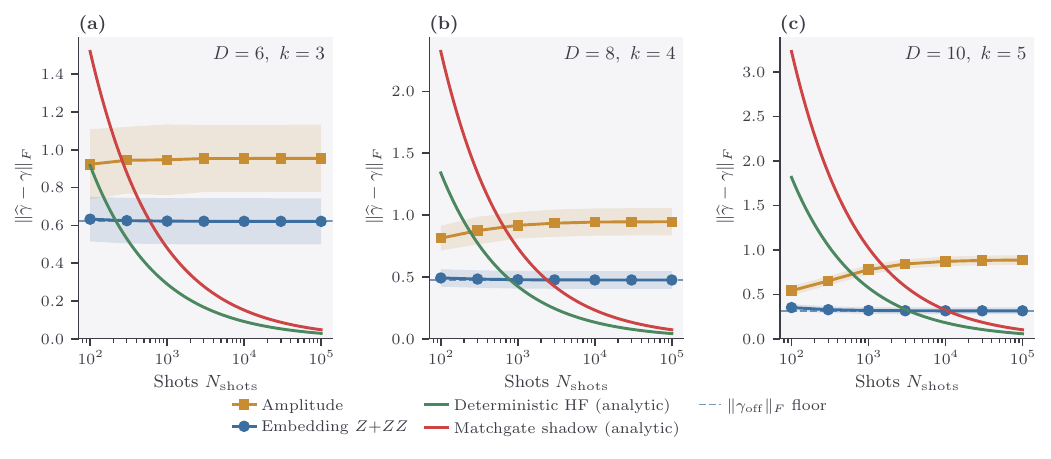}
\caption{\textbf{$1$-RDM Frobenius recovery on Haar-random states across measurement strategies.} $\norm{\widehat\gamma - \gamma}_{F}$ versus shot budget at $(D,k) = (6,3), (8,4), (10,5)$. Amplitude (squares) and $Z{+}ZZ$ (circles) are sampled from computational-basis shots on $60$ Haar-random real states in the Hamming-weight-$k$ subspace (mean, shaded band $\pm$ one standard deviation). $Z{+}ZZ$ estimates $\widehat\gamma = \mathrm{diag}(\hat n_p)$ and converges to the off-diagonal mass $\|\gamma_{\rm off}\|_F$ (dashed; $0.62$, $0.48$, $0.32$). Amplitude estimates $\widehat\gamma = \gamma(\sqrt{\hat p})$, which carries no sign information, and settles at $0.95$, $0.95$, $0.89$, above that floor. The HF and shadow curves are analytic, not simulated: $\sqrt{D(D{+}1)/2\cdot v/N_{\rm shots}}$ with per-entry single-shot variance $v = \lfloor D/2\rfloor{+}1$ (HF) and $v = 2D{-}1$ (shadow). They fall below the off-diagonal floor at $2.2\times10^{2}$, $8.0\times10^{2}$, $3.3\times10^{3}$ shots (HF) and $6.0\times10^{2}$, $2.4\times10^{3}$, $1.1\times10^{4}$ shots (shadow).}
\label{fig:gamma_unified}
\end{figure}

\paragraph{Task-level instantiation (TSP).} The framework-level results above ($O(D^3/\varepsilon^2)$ shots for full $\gamma$, $\sim 3{\times}$ HF/shadow prefactor gap at $D\le 10$, comp-basis recovery confined to $\mathrm{diag}(\gamma)$) are independent of any specific task. The remaining paragraphs of this note specialise to the TSP instantiation of \mainref{ssec:results_tsp} and report visible task-level error metrics (test BCE, tour ratio) defined there.

We compare the same protocols on the TSP-$5$ QGNNs of the full regime of \suppref{si:tsp_ablations}, trained with its earlier protocol (seeds $42$, $456$ and $1024$; validation BCE about $0.18$), sweeping the shot budget at $(D,k){=}(6,3)$ with parameter-matched classical heads ($\sim$5{,}350 parameters each). Each head is trained on exact features for $100$ epochs and evaluated at four shot budgets with a shot-noise model: each feature receives Gaussian noise of the variance of its estimator, $(1 - \langle Z_iZ_j\rangle^2)/N_{\rm shots}$ for the correlators, $(D/2+1)/N_{\rm shots}$ per $1$-RDM entry for deterministic HF and $(2D-1)/N_{\rm shots}$ for the matchgate shadow. The measurement circuits are not simulated. With exact features the HF and shadow rows coincide, since both estimate the same $1$-RDM. These readout-study rows use single-start beam decoding.

\begin{table}[h]
\centering\small
\caption{\textbf{Tour quality of the trained QGNN under each measurement strategy.} Test BCE and tour ratio at $(D,k){=}(6,3)$ on the TSP-$5$ QGNNs of \suppref{si:tsp_ablations}, three seeds, parameter-matched classical heads ($\sim 5,350$ parameters per strategy). Finite-shot columns add Gaussian noise of the analytic estimator variance to the exact features. Dashes: data not collected at this shot budget.}
\label{tab:supp_panel_b}
\begin{tabular}{@{}lc cc cc cc cc@{}}
\toprule
& & \multicolumn{2}{c}{Exact} & \multicolumn{2}{c}{$5{\times}10^4$ shots} & \multicolumn{2}{c}{$10^4$ shots} & \multicolumn{2}{c}{$10^3$ shots} \\
\cmidrule(lr){3-4}\cmidrule(lr){5-6}\cmidrule(lr){7-8}\cmidrule(lr){9-10}
Strategy & Feat. & BCE & Tour & BCE & Tour & BCE & Tour & BCE & Tour \\
\midrule
Amplitude          & $20$ & $0.176$ & $1.003$ & --- & --- & --- & --- & --- & --- \\
Embedding $Z{+}ZZ$ & $21$ & $0.391$ & $1.043$ & $0.392$ & $1.044$ & $0.396$ & $1.044$ & $0.436$ & $1.054$ \\
Deterministic HF   & $36$ & $0.239$ & $1.006$ & $0.247$ & $1.006$ & $0.281$ & $1.008$ & $0.643$ & $1.043$ \\
Matchgate shadow   & $36$ & $0.239$ & $1.006$ & $0.262$ & $1.007$ & $0.355$ & $1.014$ & $1.247$ & $1.099$ \\
\midrule
\multicolumn{10}{l}{\textit{Structural sanity check (node-register only)}} \\
Node $\langle Z_iZ_j\rangle$ & $5$ & $0.574$ & $1.104$ & $0.582$ & $1.110$ & $0.609$ & $1.126$ & $0.829$ & $1.177$ \\
\bottomrule
\end{tabular}
\end{table}

\autoref{tab:supp_panel_b} makes the information-content gap of \autoref{tab:supp_cost} visible at the task level. The embedding $Z{+}ZZ$ row changes little with the shot budget (tour ratio $1.043$ with exact features, $1.054$ at $10^3$ shots), consistent with the absence of a randomised channel inverse, but its exact-feature optimality gap is $7.4$ times that of the matchgate readouts. The deterministic HF and matchgate shadow rows reach exact-feature tour ratio $1.006$, within $0.3\%$ of the amplitude reference $1.003$, and degrade gracefully with shot budget: in this noise model HF beats the shadow at every finite-shot budget by the prefactor gap of \autoref{tab:supp_cost}, with BCE $0.643$ vs $1.247$ at $10^3$ shots, $0.281$ vs $0.355$ at $10^4$ shots, $0.247$ vs $0.262$ at $5{\times}10^4$ shots, all converging to the exact-feature BCE $0.239$ as the budget grows. The shadow row degrades sharply at $10^3$ shots, where the channel-inverse prefactor $(D+2)/2 = 4$ at $D = 6$ amplifies low-shot variance. The node-register $\langle Z_iZ_j\rangle$ row below measures only city-register correlators and ignores the embedding entirely; its task error does not improve at any shot budget. Hardware noise modelling is absent throughout.

\section{\RDM{} Jacobian rank derivation}\label{si:rdm_jacobian}

The matchgate-shadow and deterministic Hartree--Fock protocols of \suppref{si:measurement} estimate the $1$-RDM $\gamma$ on the embedding register, a $D\times D$ real symmetric matrix with trace $k$, in place of the full $\binom{D}{k}$-dimensional state. We ask what fraction of the state-amplitude information survives this compression. The answer is the generic rank of the Jacobian of the map $\Phi: \ket\psi \mapsto \gamma(\psi)$, which we derive in \autoref{prop:rdm_jacobian}: the rank is the minimum of source and target dimensions, with one additional unit lost at half-filling $k = D/2$ when source and target dimensions are comparable.

Let $\mathcal{S}_{D,k}$ denote the unit sphere of real-amplitude states in the particle-number-$k$ subspace, $\mathcal{S}_{D,k}=\{\psi\in\mathbb{R}^{\binom{D}{k}}:\|\psi\|_2=1\}$, of dimension $\binom{D}{k}-1$. Let $\mathcal{T}_D^k$ denote the affine subspace of real symmetric $D{\times}D$ matrices with trace $k$, $\mathcal{T}_D^k=\{M\in\mathrm{Sym}(\mathbb{R}^D):\Tr M=k\}$, of dimension $D(D+1)/2-1$. The map $\Phi:\mathcal{S}_{D,k}\to\mathcal{T}_D^k$ is the assignment $\psi\mapsto\gamma(\psi)$ with entries $\gamma_{pq}(\psi)=\bra{\psi}a_p^\dagger a_q\ket{\psi}$, quadratic in $\psi$. The image $\Phi(\mathcal{S}_{D,k})\subset\mathcal{T}_D^k$ is closed semialgebraic, and we ask for its generic dimension.

The differential $d\Phi|_\psi:T_\psi\mathcal{S}_{D,k}\to T_{\gamma(\psi)}\mathcal{T}_D^k$ acts between spaces of dimensions $\binom{D}{k}-1$ and $D(D+1)/2-1$ respectively. Quadratic maps between real algebraic varieties have generic rank equal to the minimum of source and target dimensions; this is Sard's theorem applied to a polynomial map plus constancy of rank on a Zariski open set. Hence
\begin{equation}
\mathrm{rank}\bigl(d\Phi|_\psi\bigr)\;\le\;\min\!\left(\tbinom{D}{k}-1,\;\tfrac{D(D{+}1)}{2}-1\right)
\label{eq:rank_naive}
\end{equation}
at every $\psi$, with equality on a generic point unless an algebraic constraint cuts the image.

At half-filling $k=D/2$, source and target are both indexed by $k$-subsets of $[D]$, and the Hodge star induces an involution on real amplitudes. Define $\star:\mathbb{R}^{\binom{D}{k}}\to\mathbb{R}^{\binom{D}{k}}$ by $(\star\psi)_S=\mathrm{sgn}(\sigma_{S,S^c})\,\psi_{S^c}$, where $\sigma_{S,S^c}$ is the permutation taking the concatenation $(S,S^c)$ into the identity ordering on $[D]$. Then $\star^2=(-1)^{k(D-k)}I$, with $\star^2=+I$ when $k(D-k)$ is even and $\star^2=-I$ otherwise. The Hodge identity for the $1$-RDM is
\begin{equation}
\gamma(\star\psi)\;=\;I-\gamma(\psi),
\label{eq:hodge_identity}
\end{equation}
a Hartree--Fock-style relation whose proof reduces to a sign computation: under $\star$, the role of $a_p^\dagger a_q$ on $\ket{S}$ with $q\in S$, $p\notin S$ is exchanged with $a_q^\dagger a_p$ on $\ket{S^c}$ with $p\in S^c$, $q\notin S^c$; the JW signs accumulated through the involution combine with the quadratic structure of $\gamma$ to give $\gamma(\star\psi)+\gamma(\psi)=I$ entrywise. We verified Eq.~\eqref{eq:hodge_identity} numerically at $(D,k)\in\{(4,2),(6,3),(8,4)\}$, observing residuals $\|\gamma(\star\psi)-(I-\gamma(\psi))\|_F$ at most $1.9\times 10^{-15}$ over $1{,}000$ random $\psi$ per case on the unit sphere.

The identity~\eqref{eq:hodge_identity} is a global property of the image $\Phi(\mathcal{S}_{D,k})$. Its consequence for the local Jacobian rank depends on whether the source and the target dimensions are comparable. When $\binom{D}{k}>D(D+1)/2$, the differential is generically full target rank and the Hodge identity merely constrains the image to a proper subset of $\mathcal{T}_D^k$ without further reducing its local dimension. When $\binom{D}{k}$ is within one of $D(D+1)/2$ at half-filling, source and target are both close to $\binom{D}{k}-1$, and the Hodge constraint cuts the image dimension by exactly one direction at generic $\psi$. Among the cases relevant to our paper, this regime contains exactly one value: $(D,k){=}(6,3)$, where $\binom{6}{3}=20$ and $D(D+1)/2=21$.

\begin{siproposition}[Generic Jacobian rank of $\Phi$]
\label{prop:rdm_jacobian}
For $\psi$ on a Zariski open subset of the unit sphere $\mathcal{S}_{D,k}$,
\begin{equation}
\mathrm{rank}\bigl(d\Phi|_\psi\bigr)
\;=\;
\min\!\left(\tbinom{D}{k}-1,\;\tfrac{D(D{+}1)}{2}-1\right)
\;-\;
\mathbf{1}\!\left[k=\tfrac{D}{2}\;\text{and}\;\tfrac{D(D{+}1)}{2}-1\le\tbinom{D}{k}\le\tfrac{D(D{+}1)}{2}\right].
\label{eq:rank_formula}
\end{equation}
The Hodge correction $-1$ at half-filling is a single real direction in the image, traced back through Eq.~\eqref{eq:hodge_identity} to the Hodge star acting on amplitudes.
\end{siproposition}

\begin{proof}
The naive bound~\eqref{eq:rank_naive} holds for any quadratic map by dimension counting. Generic equality on a Zariski open set follows from polynomial rank constancy: the locus where $\mathrm{rank}(d\Phi)$ falls below its maximum over $\mathcal{S}_{D,k}$ is a proper algebraic subvariety, hence measure zero. For the Hodge correction at half-filling, differentiate Eq.~\eqref{eq:hodge_identity} at $\psi$ in the direction $\delta\psi$ that lies along the orbit of $\star$ (that is, $\delta\psi=\star\psi-\langle\star\psi,\psi\rangle\psi$ projected onto the unit-sphere tangent); the cross-quadratic term $d\gamma|_\psi(\delta\psi)+d\gamma|_{\star\psi}(d\star\cdot\delta\psi)=0$ provides one linear constraint relating tangent vectors at $\psi$ and at $\star\psi$. When source and target dimensions are comparable (the case $D(D+1)/2-1\le\binom{D}{k}\le D(D+1)/2$), this constraint reduces the image dimension at $\psi$ by one; when the source dimension dominates the target one, the same constraint is absorbed into the rank loss already captured by the target dimension bound and produces no further deficit.
\end{proof}

We computed $\mathrm{rank}(d\Phi|_\psi)$ exactly: $\gamma$ is quadratic in $\psi$, so $d\Phi|_\psi(\delta\psi)$ has entries $2\,\psi^\top G_{pq}\,\delta\psi$ for fixed symmetric matrices $G_{pq}$. We project onto the unit-sphere tangent and count singular values above $10^{-9}$ times the largest, at five random $\psi$ per case; the rank is the same at every $\psi$.

\begin{table}[h]
\centering\small
\caption{\textbf{Jacobian rank of the embedding-to-$1$-RDM map at selected $(D,k)$.} Source dim $=\binom{D}{k}-1$; target dim $=D(D+1)/2-1$. The note column marks the unique case in our paper's range where the half-filling Hodge correction is observed.}
\label{tab:rdm_rank}
\begin{tabular}{@{}cccccl@{}}
\toprule
$D$ & $k$ & $\binom{D}{k}-1$ & $D(D{+}1)/2-1$ & rank & note \\
\midrule
$6$ & $2$ & $14$ & $20$ & $14$ & full source \\
$6$ & $3$ & $19$ & $20$ & $18$ & half-filling, Hodge $-1$ \\
$8$ & $3$ & $55$ & $35$ & $35$ & full target \\
$8$ & $4$ & $69$ & $35$ & $35$ & half-filling, source $\gg$ target \\
$10$ & $5$ & $251$ & $54$ & $54$ & half-filling, source $\gg$ target \\
\bottomrule
\end{tabular}
\end{table}

A complete sweep over $(D,k)$ with $D\le 10$ and $1\le k\le D-1$ confirms that the Hodge $-1$ correction triggers only at $(D,k){=}(6,3)$ within this range. The other half-filling cases fall outside the condition of \autoref{prop:rdm_jacobian}: at $(2,1)$ and $(4,2)$ the source dimension is well below the target and the rank equals the source dimension, and at $(8,4)$ and $(10,5)$ the rank equals the target dimension.

At $(D,k){=}(6,3)$, the Jacobian rank is $18$ out of $19$, so $\Phi$ contracts a single real direction on the tangent sphere, a $5.3\%$ information loss along that direction. The lost direction does not propagate to the downstream task: the task performance of the trained QGNN under matchgate-shadow readout (\suppref{si:measurement}, \autoref{tab:supp_panel_b}) is within $0.3\%$ of the amplitude readout on tour ratio at exact features (on the TSP-$5$ models of that table). The compression that enables polynomial-cost measurement is, at this operating point, operationally invisible.

\section{Gate-family ablation}\label{si:gate_ablation}

The trainable evolution $W$ on the embedding register (\mainref{ssec:results_quantum}) admits two ansatz families compatible with the matchgate gate set, both acting as orthogonal transformations on the particle-number-$k$ subspace of $D$ qubits. This note records their parametrisations, dynamical Lie algebras, and per-parameter gradient-variance ratio at random initialisation.

\paragraph{Nearest-neighbour pyramid ($W_{\rm NN}$).}
The pyramid of two-qubit Givens rotations introduced by~\cite{cherrat2023quantum}, written as a Cayley parametrisation of $\mathrm{SO}(D)$ on the defining $D$-dimensional representation, lifted to the particle-number-$k$ subspace by the $k$-th compound representation $C^{(k)}$. The free parameters are the $D(D{-}1)/2$ entries of an antisymmetric generator $A$, and the embedding-register action is $C^{(k)}\!\bigl((I-A)(I+A)^{-1}\bigr)$, which decomposes into nearest-neighbour Givens rotations on the $k$-subset basis. The DLA is $\mathfrak{so}(D)$ of dimension $D(D{-}1)/2$, equal to $15$ at $(D,k){=}(6,3)$. This is the choice used for the TSP numbers reported in the main text.

\paragraph{Ehrlich connectivity ($W_{\rm Ehr}$).}
The controlled-Givens particle-number encoder of~\cite{Farias2025}: a direct parametrisation of $\mathrm{SO}(\binom{D}{k})$ by a chain of single-excitation controlled-Givens rotations along the Ehrlich Gray sequence on $k$-subsets. Consecutive subsets $S$ and $S'$ in the sequence differ by a single mode swap, and one trainable angle per consecutive pair gives $\binom{D}{k}{-}1$ parameters per layer. The all-pair connectivity within the embedding register raises the DLA to $\mathfrak{so}(\binom{D}{k})$ of dimension $\binom{D}{k}(\binom{D}{k}{-}1)/2$, equal to $190$ at $(D,k){=}(6,3)$. We use this variant for trainability sweeps where a larger circuit dimension is wanted; it is not used for the numbers reported in the main text.

\paragraph{Gradient-variance comparison.}
We measure the per-parameter gradient variance of each variant at random initialisation, with the training protocol of Methods~M2: $(D,k){=}(6,3)$, two trainable layers with two trainable blocks per layer, quantum angles drawn from $\mathcal{N}(0,0.3^2)$, at $N{=}10$ cities, differentiating the mean edge logit. The variance is averaged over the trainable embedding-register parameters of $W$ (the compound-generator entries for the nearest-neighbour pyramid, the Ehrlich Gray angles for the Ehrlich connectivity), across $100$ random initialisations. We then train each variant for $300$ epochs on the $50{,}000$-instance training split with $5$ seeds ($42$, $123$, $456$, $789$, $1024$; batch $32$, one Adam optimiser, early stopping with patience $30$), at $5{,}813$ parameters with the pyramid and $5{,}829$ with Ehrlich connectivity, and report test BCE and tour ratio under beam-search decoding.

\begin{table}[h]
\centering
\small
\caption{\textbf{Nearest-neighbour pyramid versus Ehrlich connectivity on TSP-$5$.} Comparison at $(D,k){=}(6,3)$. DLA dimension is the closed-form value $\dim\mathfrak{g}$. Variance at initialisation is the per-parameter gradient variance averaged over the trainable embedding-register parameters at random initialisation ($100$ initialisations, $N{=}10$). Test BCE and tour ratio are reported as mean $\pm$ s.d.\ over $5$ seeds after at most $300$ training epochs.}
\label{tab:gate_ablation}
\begin{tabular}{@{}lccc@{}}
\toprule
Variant & $\dim\mathfrak{g}$ & Var at init. (per param.) & Test BCE / tour ratio \\
\midrule
Nearest-neighbour pyramid & $15$  & $9.9\times 10^{-4}$ & $0.204\,\pm\,0.007$ / $1.005$ \\
Ehrlich connectivity      & $190$ & $1.9\times 10^{-5}$ & $0.235\,\pm\,0.056$ / $1.008$ \\
\bottomrule
\end{tabular}
\end{table}

\paragraph{Result.}
The empirical pyramid-to-Ehrlich variance ratio is $52.6$, larger than the $\approx\!13$ predicted by the DLA-dimension ratio $\binom{D}{k}(\binom{D}{k}{-}1)/(D(D{-}1)){=}190/15$ alone. The closed-form bound of \suppref{si:trainability} is monotone decreasing in $\dim\mathfrak{g}$, so the dimension-only prediction is a lower bound on the gap; the additional empirical factor sits in the regime that Monbroussou et al.~\cite{monbroussou2025trainability} attribute to orbit structure within the particle-number subspace, which adds to trainability beyond circuit dimension. Test performance favours the nearest-neighbour pyramid at nearly matched parameter budget. One of the five Ehrlich runs stalls at test BCE $0.347$; the other four lie at $0.200$ to $0.217$, close to the pyramid's $0.195$ to $0.215$.

\paragraph{Implication.}
The nearest-neighbour pyramid has the smaller DLA and the larger gradient floor, and is the choice used in the paper on those grounds. The Ehrlich connectivity is retained as an exploratory alternative for trainability sweeps and for the $j$-WL implementations of \mainref{ssec:results_expressivity} and \suppref{si:expressivity}, where the all-pair connectivity is wanted to span the full particle-number subspace. Both variants are matchgate-compatible and so admit the polynomial-cost readouts of \mainref{ssec:results_trainability}.

\section{Frozen-parameter ablation}\label{si:tsp_ablations}

The two-register architecture combines a classical encoder and readout head with a trainable Givens evolution on the embedding register. We isolate the contribution of each by training the same model under three parameter-update regimes and comparing the downstream test metrics on a single configuration: the TSP-5 instance of Methods~M1 with $(D,k){=}(6,3)$, $L{=}2$ trainable iterations, and the readout MLP head used in \mainref{ssec:results_tsp}. The model has $5{,}813$ trainable parameters in total, of which $115$ ($2.0\%$) are quantum Givens angles and $5{,}698$ ($98.0\%$) are classical encoder and head weights.

We define three regimes. The \emph{full} regime updates all $5{,}813$ parameters; as in the main-text TSP-5 model. The \emph{quantum-only} regime freezes the classical encoder weights and the readout MLP head at their initial values and updates only the $115$ Givens angles. The \emph{classical-only} regime freezes the Givens angles at their zero-mean Gaussian initialisation ($\sigma{=}0.3$ rad, M2) and updates the $5{,}698$ classical parameters. Each regime is run with five seeds $\{42, 123, 456, 789, 1024\}$, and within each seed, the optimiser state is reset before training so that the frozen subset is fixed at the chosen initialisation. All regimes use one Adam optimiser with weight decay and cosine annealing, batch size $32$, the class-weighted binary cross-entropy and early stopping with patience $30$, and train on the $50{,}000$ generated TSP-5 training instances. This earlier protocol differs from the main-text runs (\suppref{si:tsp_benchmark}) in batch size, patience and training-set size. Test BCE and tour ratio are reported on the held-out generated test split.

\autoref{tab:tsp5_frozen} reports the resulting metrics. The full model attains test BCE $0.191 \pm 0.014$ and tour ratio $1.004 \pm 0.001$ across the five seeds. Freezing the classical encoder and head reduces the trainable parameter count from $5{,}813$ to $115$, and the test BCE rises to $0.585 \pm 0.011$ with a tour ratio of $1.059 \pm 0.009$. Freezing the quantum Givens evolution reduces the trainable parameter count from $5{,}813$ to $5{,}698$, and the test BCE rises to $0.411 \pm 0.014$ with tour ratio $1.026 \pm 0.005$. Both frozen regimes degrade the metrics relative to the full model on both BCE and tour ratio.

\begin{table}[ht]
\centering
\footnotesize
\caption{\textbf{Frozen-parameter ablation on TSP-$5$.} Three parameter-update regimes at the configuration of Methods~M1: full, quantum-only, and classical-only. Five seeds per regime; mean $\pm$ standard deviation. Trainable parameter counts refer to the parameters updated by gradient descent in each regime; the model itself has $5{,}813$ parameters in all three.}
\label{tab:tsp5_frozen}
\begin{tabular}{@{}lccc@{}}
\toprule
Regime & Trainable params & Test BCE & Tour ratio \\
\midrule
Full              & $5{,}813$ & $0.191 \pm 0.014$ & $1.004 \pm 0.001$ \\
Quantum-only      & $115$     & $0.585 \pm 0.011$ & $1.059 \pm 0.009$ \\
Classical-only    & $5{,}698$ & $0.411 \pm 0.014$ & $1.026 \pm 0.005$ \\
\bottomrule
\end{tabular}
\end{table}

The two frozen regimes give a one-sided check on the contribution of each sub-block. The classical-only regime keeps an untrained Givens evolution at random initialisation, so the gradient signal that reaches the encoder and head must propagate through a fixed quantum map; the resulting BCE of $0.411$ shows that a substantial fraction of the task fit is recoverable through the classical sub-block alone, but a gap of $0.220$ in BCE and $0.022$ in tour ratio remains relative to the full model. The quantum-only regime keeps the encoder and head at their initialisation, so the only signal available to the model passes through $115$ trainable Givens angles; the resulting BCE of $0.585$ is close to the chance-level binary cross-entropy of $\log 2 \approx 0.69$ scaled by the class-balance re-weighting used in M2, indicating that $115$ Givens angles in isolation are too few to fit the edge-prediction target on this task. The full model is the only regime that reaches a tour ratio below $1.01$. Hence, the encoder and head supply the majority of the parametric capacity needed for the BCE fit, and the trainable Givens evolution supplies the remaining gap; neither sub-block reaches the full-model metrics on its own at this parameter budget. The result fixes the role of the quantum evolution within the architecture as a bottleneck layer that the classical sub-blocks rely on, and not as a stand-alone learner of the TSP edge distribution; it does not bear on quantum advantage, which is addressed separately in \mainref{ssec:results_trainability}.

\section{Resource estimates and gate decompositions}\label{si:resources}

This note gives the per-block decomposition behind the resource estimate of Methods~M3. We count the gates of two trained TSP configurations at $(D, k){=}(6, 3)$ and node-register particle number $j{=}1$: $N{=}20$ with $L{=}1$ layer and $N{=}50$ with $L{=}2$, both with four trainable pyramids per layer (\autoref{tab:si_tsp_hparams}). The counts follow the simulated model layer by layer and are printed by \texttt{count\_gates.py} in the code repository. The two registers occupy $N + D$ qubits, equal to $26$ at $N{=}20$ and $56$ at $N{=}50$.

\paragraph{Givens decomposition.} A particle-number-preserving Givens rotation $G_{ab}(\theta)$ acts on two qubits $a, b$ as a real $\mathrm{SO}(2)$ rotation in the two-dimensional subspace $\mathrm{span}\{|01\rangle, |10\rangle\}$. On a \textsc{cnot}-native instruction set the canonical decomposition of Anselmetti et al.~\cite{Anselmetti2021} gives
\begin{equation}
G_{ab}(\theta) \;=\; 2 \,\textsc{cnot} \;+\; 2\, R_y(\pm \theta/2),
\label{eq:givens_count}
\end{equation}
i.e.\ two two-qubit \textsc{cnot}s and two single-qubit $y$-rotations per Givens, with a fixed Hadamard frame on either side. The decomposition is exact up to a global phase and uses no ancillas. Controlled Givens rotations, used inside the hierarchical loader for cross-register connections that thread the Ehrlich sequence onto the embedding qubits conditional on the active node, take a constant prefactor more under SWAP routing on a linearly-connected device: each controlled Givens decomposes to a sandwich of two \textsc{cnot}s around a controlled $R_y$, which itself unfolds into two further \textsc{cnot}s and two $R_y$ on a \textsc{cnot}-native gate set, plus the SWAP overhead to bring the control qubit adjacent to the rotation pair. We absorb this prefactor into the per-block count below; it does not change the leading-order scaling.

\paragraph{Per-block accounting.} Each layer loads the coordinates, applies the equivariant adjacency $\mathcal{A}$ and then the trainable evolution $W$; the joint mixing $M$ is applied once, after the last layer, followed by the measurement basis change. We count primitive Givens per block.

\textit{Node register.} On the first layer, $N{-}1$ Givens prepare the uniform one-city state $\frac{1}{\sqrt N}\sum_i|i\rangle$ on the node register; every later layer applies the same $N{-}1$ Givens again before it reloads.

\textit{Hierarchical loader $V$.} Each city's coordinates set the angles of the Ehrlich Gray-code chain on the particle-number-$k$ subspace, $\binom{D}{k}-1 = 19$ Givens, each controlled on that city's node qubit. A layer therefore applies $19N$ controlled Givens: $380$ at $N{=}20$ and $950$ at $N{=}50$.

\textit{Equivariant adjacency $\mathcal{A}$.} The adjacency layer applies one Givens per pair of nodes in the canonical edge ordering, i.e.\ $\binom{N}{2}$ Givens per layer: $190$ at $N{=}20$ and $1{,}225$ at $N{=}50$.

\textit{Trainable evolution $W$.} The nearest-neighbour pyramid parametrises a Cayley element of $\mathrm{SO}(D)$ and uses $D(D{-}1)/2 = 15$ Givens at $D{=}6$. The compound lift $C^{(k)}$ is implemented at the matrix-multiplication level on the embedding-register state and is not itself realised by additional gates: the $15$ base-register Givens generate the $\binom{D}{k}{=}20$-dimensional compound action through the natural action of $\mathrm{SO}(D)$ on the particle-number-$k$ subspace, with depth $O(D)$ on a linearly-connected device. With four pyramids per layer this is $60$ Givens per layer.

\textit{Joint mixing $M$.} A nearest-neighbour pyramid of Givens over all $N{+}D$ qubits, applied once: $\binom{N+D}{2}$ Givens, $325$ at $N{=}20$ and $1{,}540$ at $N{=}50$. In the simulated model a learnable per-component scaling of the amplitudes precedes it; that step has no circuit form and is not counted.

\textit{Measurement basis change.} The deterministic Hartree--Fock $1$-RDM readout uses $D/2 + 1 = 4$ measurement settings, each a single layer of Givens on the embedding register at depth $O(D)$, contributing a constant additive cost independent of $N$ and $L$. The matchgate-shadow alternative uses one Haar matchgate per snapshot at depth $O(D^2) = 36$. Both are negligible relative to the loader and adjacency costs.

\begin{table}[ht]
\centering
\footnotesize
\begin{tabular}{@{}lccc@{}}
\toprule
Block & Givens per layer & $N{=}20$, $L{=}1$ & $N{=}50$, $L{=}2$ \\
\midrule
Node register (prepare, reload) & $N-1$ & $19$ & $98$ \\
Hierarchical loader $V$ (controlled) & $19N$ & $380$ & $1{,}900$ \\
Equivariant adjacency $\mathcal{A}$ & $\binom{N}{2}$ & $190$ & $2{,}450$ \\
Trainable evolution $W$ ($4$ pyramids) & $60$ & $60$ & $120$ \\
Joint mixing $M$ (once) & --- & $325$ & $1{,}540$ \\
\midrule
Total per forward pass & --- & $974$ & $6{,}108$ \\
\bottomrule
\end{tabular}
\caption{\textbf{Per-block Givens count of the trained TSP circuits at $(D, k, j) = (6, 3, 1)$.} The measurement basis change adds $D/2+1=4$ settings, each one layer of Givens on the embedding register. Counts from \texttt{count\_gates.py}.}
\label{tab:resource_blocks}
\end{table}

\paragraph{Total per forward pass.} A single forward pass uses $974$ primitive Givens rotations at $N{=}20$ and $6{,}108$ at $N{=}50$. By~\eqref{eq:givens_count}, a plain Givens compiles to $2$ \textsc{cnot}s; a controlled Givens of the loader needs about $4$, before any routing overhead. This gives about $2{,}700$ \textsc{cnot}s at $N{=}20$ and $16{,}000$ at $N{=}50$, plus the single-qubit rotations and the constant Hadamard frame. The growth with $N$ comes from the $19N$ loader gates, the $\binom{N}{2}$ adjacency gates and the $\binom{N+D}{2}$ gates of the joint mixer. The model hyperparameters $(D, k, j)$ are chosen independently of $N$, and the per-layer trainable feature-embedding evolution does not change with $N$; the joint mixer adds $O(N^2)$ trainable angles (Methods~M2).

\paragraph{Hardware feasibility.} The $26$-qubit configuration at $N{=}20$ and the $56$-qubit deployment target at $N{=}50$ sit inside the qubit envelope of present-day superconducting platforms: Google's Willow processor at $\sim 105$ qubits~\cite{WillowNature2025} and IBM Heron at $\sim 156$ qubits. The deployment claim attached to these numbers is functional rather than complexity-theoretic: the trained parameters can be evaluated on matchgate-compatible hardware at this scale. Noise tolerance under realistic two-qubit gate error rates and shot budgets is not characterised in the present paper, and the resource estimate above assumes ideal gates in the abstract gate-count sense. This note is the per-block backing for the numbers cited in Methods~M3.

\end{document}